\documentclass[draft,onecolumn]{IEEEtran}
\IEEEoverridecommandlockouts
\usepackage{cite}
\usepackage{amsmath,amssymb,amsfonts,amsthm}
\usepackage{algorithmic}
\usepackage{graphicx}
\usepackage{textcomp}
\usepackage{xcolor}
\usepackage{multirow}
\usepackage{bm}
\usepackage{setspace} 
\usepackage{tikz}
 \usetikzlibrary{arrows.meta,positioning}
\usepackage{url}
\usepackage[table]{xcolor}
\usetikzlibrary{positioning}
\usetikzlibrary{calc,positioning}
\usepackage{makecell}
\usepackage{subcaption}

\usepackage{blkarray}
\usepackage{comment}

\newif\ifcomment
\commenttrue

\begin{document}

\newtheorem{theorem}{Theorem}
\newtheorem{proposition}{Proposition}
\newtheorem{definition}{Definition}
\newtheorem{corollary}{Corollary}
\newtheorem{lemma}{Lemma}
\newtheorem{example}{Example}
\newtheorem{remark}{Remark}
\newtheorem{construction}{Construction}
\newtheoremstyle{lpstyle}{}{}{\itshape}{}{\bfseries}{.}{.5em}
  {\thmname{#1}\thmnumber{#2}\thmnote{ (#3)}}
\theoremstyle{lpstyle}
\newtheorem{linearprogram}{LP}
\theoremstyle{plain}

\title{Generalized Function-Correcting Partition Codes: Bounds and constructions}

\author{\IEEEauthorblockN{Charul Rajput$^*$, Mahak$^\dagger$, V. Lalitha$^\dagger$}\\
\IEEEauthorblockA{$^*$University of Bristol Mumbai Enterprise Campus, India\\
$^\dagger$Signal Processing and Communication Research Center (SPCRC),\\
International Institute of Information Technology Hyderabad, India\\
Emails: $^*$charul.rajput@mumbai.bristol.ac.uk, $^\dagger$\{mahak@research., lalitha.v@\}iiit.ac.in}}

\maketitle

\begin{abstract}
We introduce generalized function-correcting partition codes (GFCPCs) that simultaneously protect multiple partitions of the message space against different numbers of errors. Given partitions with respective distance requirements, a GFCPC is a systematic encoding that guarantees, for each partition, a specified minimum Hamming distance between codewords whose messages lie in different blocks. This framework unifies and generalizes both function-correcting partition codes, which protect multiple functions with a common error-correction level, and function-correcting codes with data protection, which assign different levels of protection to data and a single function. We define the distance requirement matrix $\mathcal{D}$ for GFCPCs and use it to characterize the optimal redundancy in terms of the shortest length of an associated $\mathcal{D}$-code. We derive general upper and lower bounds on the optimal redundancy, including an upper bound that considers the join of different combinations of the partitions. For two partitions of the message space over the binary field, we establish improved lower bounds on the optimal redundancy under specific neighborhood conditions on the partitions. We develop a linear programming bound on the optimal redundancy of GFCPCs. The bound specializes to function-correcting codes and function-correcting codes with data protection, for which, to the best of our knowledge, no such bound was previously known. We present a multi-step construction procedure for these codes and illustrate it with some examples. Through several examples, we demonstrate that the proposed framework can yield strictly smaller redundancy than both the sum of the individual FCPC redundancies and the redundancy of a single FCPC designed for the join partition with the highest distance (strongest protection required).
\end{abstract}
%

 
\section{Introduction}

In a standard communication or storage system, classical error-correcting codes are designed to protect the entire message against errors. However, in many practical scenarios, the receiver is not interested in recovering the full message but only in computing a specific function of it. For instance, in distributed computing, intermediate or terminal nodes may only need aggregate quantities such as sums, maxima, or classification labels derived from the data. Protecting the entire message in such settings introduces unnecessary redundancy. Function-correcting codes (FCCs) were introduced by Lenz et al. in \cite{LBZY2023} to address this, and protect only the value of a desired function against errors. The main idea behind FCCs is that codewords corresponding to messages with the same function value need not be separated by any minimum distance, since confusing them does not affect the receiver's computation. Distance constraints are imposed only between codewords whose messages correspond to different function values. This targeted protection allows FCCs to achieve lower redundancy compared to classical error-correcting codes that protect the full message.

Since a function naturally partitions its domain into preimage sets, the encoding of an FCC depends only on the partition induced by the function, not on the specific function itself. This observation led to the introduction of function-correcting partition codes (FCPCs) in \cite{RRHH20261}. An FCPC is defined directly on a partition of the message space, and any FCC for a function $f$ is exactly an FCPC with respect to the partition induced by $f$. This abstraction offers two advantages. First, a single FCPC can simultaneously serve all functions that induce the same domain partition, providing a notion of structural privacy since the transmitter need not know the exact function desired by the receiver. Second, when multiple functions need to be protected, FCPCs employ the join of the corresponding domain partitions to construct a single code that protects all functions simultaneously with a common error-correction level, potentially saving redundancy compared to using separate codes for each function.

A separate but related direction was pursued in \cite{RRHH20251}, where FCCs with data protection were introduced. In this framework, the encoding provides two levels of error protection, a baseline level for the data itself and a stronger level for the function value. This is motivated by network scenarios where reliable data flow is essential even though certain attributes of the data require stronger protection. The work in \cite{RRHH20251} proposed a two-step construction, derived bounds on the optimal redundancy, and showed that in some cases data protection can be added without increasing the redundancy beyond that of the original FCC.

However, both of these frameworks have limitations when considered individually. FCPCs in \cite{RRHH20261} allow protection of multiple functions but only with a common error-correction capability. FCCs with data protection in \cite{RRHH20251} allow two different levels of protection but only for a single function alongside the data. In many practical settings, one may need to protect multiple functions or attributes of the data, each requiring a different level of error protection. For example, in a data storage system, a coarse classification of the stored data may need only moderate protection, while a finer attribute may require strong protection, and the raw data itself may need only basic protection. Similarly, in a broadcast setting with multiple receivers, different receivers may require protection for different functions at different reliability levels.

To address this, we propose generalized FCPCs, a framework in which multiple partitions of the message space are simultaneously protected, each with its own distance requirement. This framework subsumes both FCCs with data protection and FCPCs as special cases. We study the optimal redundancy of these codes through combinatorial bounds, a linear programming bound, and a multi-step construction method.

\subsection{Related Work}

FCCs were formally introduced in \cite{LBZY2023}, where the authors established their equivalence to irregular-distance codes and derived general upper and lower bounds on the optimal redundancy. They also provided constructions for specific function families, including locally binary functions, the Hamming weight function, and the Hamming weight distribution function. The work in \cite{PR2024} derived new lower bounds on the redundancy of FCCs and studied FCCs for linear functions. Ge et al. in \cite{GXZZ2025} obtained improved bounds on the optimal redundancy for Hamming weight and Hamming weight distribution functions.

The FCC framework has also been extended to different channel models. Xia et al. \cite{XLC2024} studied FCCs for symbol-pair read channels, and Singh et al. \cite{SSY2025} generalized FCCs to $b$-symbol read channels over finite fields. The work in \cite{RRHH20252} introduced the class of locally $(\rho, \lambda)$-bounded functions and developed redundancy bounds for these functions. Further,  \cite{VSA2025} introduced locally $(\lambda, \rho, b)$-functions and investigated the optimal redundancy of Function-Correcting $b$-Symbol Codes (FCBSCs). In addition, \cite{SR2025} studied FCCs for linear functions in the $b$-symbol read channel and derived a Plotkin-like bound, which reduces to the corresponding bound for symbol-pair channels when $b = 2$. Ly and Soljanin in \cite{LS2025} derived bounds on the redundancy of FCCs over finite fields and proposed encoding schemes achieving optimal redundancy for sufficiently large fields. More recently, non-existence results for strict FCCs with data protection were established in \cite{RRHH20262} using a graph-theoretic framework based on the distance graphs of codes.

Several variants and generalizations of FCCs have been studied. In \cite{RRHH20251}, a framework was proposed to simultaneously protect both the function value $f(u)$ and the message $u$ with different levels of error protection. The concept of Function-Correcting Partition Codes (FCPCs), based on partitions of the message space rather than a specific function, was introduced in \cite{RRHH20261}.

In addition to the Hamming metric, FCCs have been investigated under various metrics and algebraic structures. In \cite{HH2026}, FCCs with homogeneous distance were introduced over $\mathbb{Z}_{2^l}~(l\geq 2)$, where the homogeneous metric coincides with the Lee metric over $\mathbb{Z}_4$. This line of research was extended in \cite{VS2025} to codes over $\mathbb{Z}_m~(m\geq 2)$ under the Lee metric, focusing on optimal redundancy. Further generalization to the chain ring $\mathbb{Z}_{2^s}$ was carried out in \cite{VS2026}, covering broader classes of functions such as locally bounded functions, linear functions, weight function, and modular sum functions. Additionally, \cite{HUR2025} proposed a Plotkin-like bound for irregular Lee-distance codes and constructed explicit Function-Correcting Lee Codes (FCLCs) for several function classes, including Lee weight, Lee weight distribution, modular sum, and locally bounded functions.

In \cite{SS2025}, the authors introduced function-correcting insertion, deletion, and insertion–deletion (insdel) codes (FCIDCs) and showed the equivalence of these formulations. Finally, in \cite{PBMR2026}, optimal Single-Error Correcting FCCs (SEFCCs) were studied for maximally unbalanced Boolean functions, while \cite{DMKPR2026} constructed SEFCCs for the Hamming Code Membership Function (HCMF), which determines whether a vector in $\mathbb{F}_2^7$ belongs to the $[7, 4, 3]$-Hamming code.

The bounds available for FCCs and their variants are combinatorial in nature. For classical error-correcting codes, the linear programming bound of Delsarte \cite{D1973} gives a well-known general upper bound on the size of a code with a prescribed minimum distance. The bound is obtained from the distance distribution of the code together with the MacWilliams inequalities, and is treated in detail in \cite[Ch.~17]{MS1977}. No linear programming bound has previously been given for FCCs, for FCCs with data protection, or for FCPCs. The distance constraints in these settings depend on the partition blocks of the messages, so the classical distance distribution does not capture the relevant structure.

\subsection{Contributions}

The main contributions of this paper are as follows.
\begin{enumerate}
    \item We introduce the notion of $(\mathcal{P}_h : d_h;\, h \in [H])$-FCPCs, which simultaneously protect $H$ partitions of the message space with respective distance requirements. We show that this framework generalizes both FCPCs for multiple functions \cite{RRHH20261} and FCCs with data protection \cite{RRHH20251}.

    \item We define the distance requirement matrix (DRM) for generalized FCPCs and use it to characterize the optimal redundancy in terms of the shortest length of an associated $\mathcal{D}$-code. This provides a tool for computing lower bounds by restricting attention to suitably chosen subsets of messages.

    \item We derive general lower and upper bounds on the optimal redundancy of generalized FCPCs. The upper bound is obtained by minimizing over all possible groupings of the $H$ partitions, where each group is protected by a single FCPC for the join of its partitions, at the maximum distance required by any partition in that group.

    \item For the binary case with two partitions, we establish improved lower bounds on the optimal redundancy under certain neighborhood conditions on the partitions.

      \item We develop a linear programming bound on the optimal redundancy of generalized FCPCs. We introduce distance distributions indexed by the subsets of partitions in which a pair of messages lies in different blocks, and derive MacWilliams-type inequalities for these distributions. The resulting linear program bounds the number of codewords of any generalized FCPC of a given length. The bound specializes to FCCs and to FCCs with data protection, for which, to the best of our knowledge, no such bound was previously known.

     \item We present an $H$-step construction method that builds the encoding iteratively. At each step, the distance guarantee is upgraded for the partitions that require stronger protection. The method fixes the structure of the encoding, and the component encodings at each step are chosen according to the given partitions. We prove the correctness of the method and illustrate it through several examples.
\end{enumerate}

The rest of the paper is organized as follows. Section~\ref{sec:prelim} recalls the necessary background on partitions, function-correcting codes, and FCPCs. Section~\ref{sec:main} introduces the generalized FCPC framework, defines the distance requirement matrix, and states the main results. Section~\ref{sec:bounds} derives combinatorial upper and lower bounds on the optimal redundancy, including lower bounds obtained from the distance requirement matrix and improved lower bounds for the binary case with two partitions. Section~\ref{sec:lp} develops the linear programming bound. Section~\ref{sec:construction} presents the multi-step construction method and demonstrates it through examples. Section~\ref{sec:conclusion} concludes the paper.

\subsection{Notation}
Let $\mathbb{F}_q$ denote the finite field with $q$ elements, and let $\mathbb{F}_q^k$ denote the vector space of all $k$-tuples over $\mathbb{F}_q$. For a positive integer $n$, we write $[n] = \{1, 2, \ldots, n\}$. For a finite set $A$, we write $|A|$ for its cardinality. The \emph{support} of a vector $u = (u_1, \ldots, u_k) \in \mathbb{F}_q^k$ is the set of coordinates at which $u$ is nonzero, denoted by $\mathrm{supp}(u) = \{ i \in [k] : u_i \neq 0 \}$. The Hamming weight of a vector $u \in \mathbb{F}_q^k$, denoted by $\mathrm{wt}(u)$, is the number of nonzero coordinates of $u$, so $\mathrm{wt}(u) = |\mathrm{supp}(u)|$. The Hamming distance between two vectors $u, v \in \mathbb{F}_q^k$ is denoted by $d(u, v)$ and equals the number of coordinates in which $u$ and $v$ differ. A vector $v \in \mathbb{F}_q^k$ is called a \emph{neighbor} of $u \in \mathbb{F}_q^k$ if $d(u, v) = 1$. An encoding $\mathcal{C} : \mathbb{F}_q^k \to \mathbb{F}_q^n$ is called \emph{systematic} if for every $u \in \mathbb{F}_q^k$, the codeword $\mathcal{C}(u)$ has $u$ as a prefix, i.e., $\mathcal{C}(u) = (u, p(u))$ where $p(u) \in \mathbb{F}_q^{n-k}$ is the redundancy (or parity) vector. The integer $r = n - k$ is the \emph{redundancy} of the encoding. The image of the encoding is the code $C = \mathrm{Im}(\mathcal{C}) = \{ \mathcal{C}(u) : u \in \mathbb{F}_q^k \}$. We write $\mathbb{N}_0$ for the set of nonnegative integers, and $[x]^+ \triangleq \max\{x, 0\}$ for a real number $x$. 


\section{Preliminaries}\label{sec:prelim}

In this section, we recall the basic definitions and results from the literature that are used in this paper.

\subsection{Partitions}

We recall the following standard definitions on partitions from~\cite{Aumann1976, Quint2014}.

\begin{definition}[Partition]
A partition $\mathcal{P}$ of a finite set $S$ is a collection of pairwise disjoint nonempty subsets of $S$ whose union is $S$. The elements of $\mathcal{P}$ are called its \emph{blocks}.
\end{definition}

\begin{definition}[Refinement]
Let $\mathcal{P}$ and $\mathcal{Q}$ be two partitions of a set $S$. The partition $\mathcal{Q}$ is called a \emph{refinement} of $\mathcal{P}$ if every block of $\mathcal{Q}$ is contained in some block of $\mathcal{P}$. In this case, we say that $\mathcal{Q}$ is \emph{finer} than $\mathcal{P}$, or equivalently, $\mathcal{P}$ is \emph{coarser} than $\mathcal{Q}$.
\end{definition}

\begin{definition}[Join]
Given two partitions $\mathcal{P}$ and $\mathcal{Q}$ of a set $S$, their \emph{join}, denoted by $\mathcal{P} \vee \mathcal{Q}$, is the coarsest partition that is a refinement of both $\mathcal{P}$ and $\mathcal{Q}$. Each block of the join is the intersection of a block from $\mathcal{P}$ and a block from $\mathcal{Q}$. More generally, for partitions $\mathcal{P}_1, \mathcal{P}_2, \ldots, \mathcal{P}_H$ of $S$, their join $\mathcal{P}_1 \vee \mathcal{P}_2 \vee \cdots \vee \mathcal{P}_H$ is the coarsest partition that refines all of $\mathcal{P}_1, \ldots, \mathcal{P}_H$.
\end{definition}

We note that the \emph{finest partition} of $S$ is $\big\{\{a\} : a \in S\big\}$, in which every block is a singleton.

For a function $f : \mathbb{F}_q^k \to S$, the \emph{domain partition induced by $f$} is the partition $\mathcal{P}_f = \{f^{-1}(s) : s \in \mathrm{Im}(f)\}$ of $\mathbb{F}_q^k$ into the preimage sets of $f$.

\subsection{Function-Correcting Codes}

Function-correcting codes were introduced in \cite{LBZY2023} to protect the value of a function against errors while using less redundancy than classical error-correcting codes.

\begin{definition}[Function-correcting code {\cite{LBZY2023}}]\label{def:fcc}
Let $f : \mathbb{F}_q^k \to S$ be a function and $t$ be a positive integer. A systematic encoding $\mathcal{C} : \mathbb{F}_q^k \to \mathbb{F}_q^{k+r}$ is called a $t$-error \emph{function-correcting code} for $f$, denoted $(f, t)$-FCC, if for all $u, v \in \mathbb{F}_q^k$ with $f(u) \neq f(v)$,
\[
d\big(\mathcal{C}(u), \mathcal{C}(v)\big) \geq 2t + 1.
\]
The optimal redundancy of an $(f, t)$-FCC, denoted $r_f(k, t)$, is the smallest integer $r$ for which such an encoding exists.
\end{definition}

The work in \cite{LBZY2023} established the equivalence between FCCs and irregular-distance codes or $\mathcal{D}$-codes, which are codes where the minimum distance requirement between pairs of codewords can vary.

\begin{definition}[$\mathcal{D}$-code {\cite{LBZY2023}}]\label{def:dcode}
Let $\mathcal{D} \in \mathbb{N}_0^{M \times M}$. A collection of $M$ vectors $p_1, p_2, \ldots, p_M \in \mathbb{F}_q^r$, not necessarily distinct, is called a \emph{$\mathcal{D}$-code} if
\[
d(p_i, p_j) \geq \mathcal{D}_{i,j}, \quad \text{for all } i, j \in [M].
\]
The minimum length of a $\mathcal{D}$-code is denoted by $N(\mathcal{D})$.
\end{definition}

The optimal redundancy of an $(f, t)$-FCC can be characterized using $\mathcal{D}$-codes. Define the \emph{distance requirement matrix} $\mathcal{D}_f(t, u_1, \ldots, u_{q^k}) \in \mathbb{N}_0^{q^k \times q^k}$ for an $(f, t)$-FCC as the matrix with entries
\[
[\mathcal{D}_f(t, u_1, \ldots, u_{q^k})]_{ij} = \begin{cases}
[2t + 1 - d(u_i, u_j)]^+, & \text{if } f(u_i) \neq f(u_j), \\
0, & \text{otherwise},
\end{cases}
\]
for some fixed ordering $u_1, u_2, \ldots, u_{q^k}$ of $\mathbb{F}_q^k$. Then the optimal redundancy satisfies $r_f(k, t) = N(\mathcal{D}_f(t, u_1, \ldots, u_{q^k}))$ \cite{LBZY2023}.

\subsection{Function-Correcting Partition Codes}

FCPCs were introduced in \cite{RRHH20261} as a generalization of FCCs, defined directly on partitions of the message space rather than on specific functions.

\begin{definition}[Function-correcting partition code {\cite{RRHH20261}}]\label{def:fcpc}
Let $\mathcal{P} = \{P^1, P^2, \ldots, P^E\}$, where $E = |\mathcal{P}|$, be a partition of $\mathbb{F}_q^k$ and $t$ be a positive integer. A systematic encoding $\mathcal{C} : \mathbb{F}_q^k \to \mathbb{F}_q^{k+r}$ is called a $(\mathcal{P}, t)$-\emph{encoding} (or $t$-error FCPC for $\mathcal{P}$) if for all $u \in P^i$ and $v \in P^j$ with $i \neq j$,
\[
d\big(\mathcal{C}(u), \mathcal{C}(v)\big) \geq 2t + 1.
\]
\end{definition}

The optimal redundancy of a $(\mathcal{P}, t)$-encoding, denoted $r_\mathcal{P}(k, t)$, is the smallest integer $r$ for which such an encoding exists. When expressed in terms of the distance requirement $d = 2t+1$, we write $r_\mathcal{P}(k : d)$ for the optimal redundancy. That is, $r_\mathcal{P}(k, t) = r_\mathcal{P}(k : 2t+1)$.

As observed in \cite{RRHH20261}, any $(f, t)$-FCC is exactly a $(\mathcal{P}_f, t)$-encoding with respect to the domain partition $\mathcal{P}_f$ induced by $f$, making FCPCs a natural generalization of FCCs. When multiple functions $f_1, f_2, \ldots, f_L$ need to be protected with a common error-correction level $t$, the work in \cite{RRHH20261} showed that a single $(\mathcal{P}_{f_1} \vee \mathcal{P}_{f_2} \vee \cdots \vee \mathcal{P}_{f_L},\, t)$-encoding simultaneously protects all functions, potentially with smaller redundancy than constructing separate FCCs for each function.

\subsection{Function-Correcting Codes with Data Protection}

FCCs with data protection were introduced in \cite{RRHH20251} to provide different levels of error protection to the data and the function value simultaneously.

\begin{definition}[FCC with data protection {\cite{RRHH20251}}]\label{def:fcc_dp}
Let $f : \mathbb{F}_q^k \to \mathrm{Im}(f)$ be a function. A systematic encoding $\mathcal{C} : \mathbb{F}_q^k \to \mathbb{F}_q^{k+r}$ is called an $(f:\, d_d, d_f)$-FCC if the following two conditions hold:
\begin{enumerate}
    \item For all $u, v \in \mathbb{F}_q^k$ with $u \neq v$,
    \[
    d\big(\mathcal{C}(u), \mathcal{C}(v)\big) \geq d_d.
    \]
    \item For all $u, v \in \mathbb{F}_q^k$ with $f(u) \neq f(v)$,
    \[
    d\big(\mathcal{C}(u), \mathcal{C}(v)\big) \geq d_f.
    \]
\end{enumerate}
Here $d_d$ is the minimum distance for data protection, and $d_f (\geq d_d)$ is the minimum distance for function-value protection. The encoding corrects up to $t_d = \lfloor(d_d - 1)/2\rfloor$ errors for the data and up to $t_f = \lfloor(d_f - 1)/2\rfloor$ errors for the function value.
\end{definition}

\begin{remark}\label{rem:dp_as_fcpc}
An $(f:\, d_d, d_f)$-FCC can be viewed as a code with two partitions: $\mathcal{P}_1$, the finest partition $\big\{\{u\}: u \in \mathbb{F}_q^k\big\}$, with distance requirement $d_d$, and $\mathcal{P}_2$, the domain partition induced by $f$, with distance requirement $d_f$. This viewpoint motivates the generalization to multiple partitions with different distance requirements that we introduce in this paper.
\end{remark}

\subsection{The Linear Programming Bound for Codes}\label{subsec:lp-prelim}
We recall the linear programming bound of Delsarte \cite{D1973} for codes over $\mathbb{F}_q$. The binary case is given in \cite[Ch.~5, Ch.~17]{MS1977}, and it was later generalized to arbitrary $q$ in \cite[Ch. 5]{vL1999}. Let $C \subseteq \mathbb{F}_q^n$ be a code with $|C| = M$. The \emph{distance distribution} of $C$ is the sequence $B_0, B_1, \ldots, B_n$ given by
\[
B_i = \frac{1}{M} \big| \{ (x, y) \in C \times C : d(x, y) = i \} \big|.
\]
Thus $B_0 = 1$, each $B_i$ is nonnegative, and $\sum_{i=0}^{n} B_i = M$. The distance distribution does not require $C$ to be linear. It is distinguished from the weight distribution $A_0, A_1, \ldots, A_n$, where $A_i$ counts the codewords of Hamming weight $i$. The two coincide when $C$ is linear.
For $0 \leq j \leq n$, the \emph{Krawtchouk polynomial} of degree $j$ is
\begin{equation}
K_j(i) = \sum_{l=0}^{j} (-1)^l (q-1)^{j-l} \binom{i}{l} \binom{n-i}{j-l}, \quad 0 \leq i \leq n.
\label{eq:krawtchouk}
\end{equation}

The distance distribution of any code over $\mathbb{F}_q$ satisfies the following inequalities.
\begin{lemma}
[{\cite[Lemma 5.3.3]{vL1999}}]\label{thm:macwilliams}
Let $C \subseteq \mathbb{F}_q^n$ be a code with distance distribution $B_0, B_1, \ldots, B_n$. Then
\[
\sum_{i=0}^{n} B_i K_j(i) \geq 0, \quad \text{for all } j \in \{0, 1, \ldots, n\}.
\]  
\end{lemma}

The result is proved in \cite{vL1999} for codes over the ring $\mathbb{Z}/q\mathbb{Z}$, and since the inequalities depend on $C$ only through its Hamming distances, they hold over any alphabet of size $q$.
The inequalities in Lemma~\ref{thm:macwilliams} are linear in the $B_i$. Since the size of the code is $\sum_{i=0}^{n} B_i$, an upper bound on $|C|$ is obtained by maximizing this sum over all nonnegative sequences satisfying the inequalities together with the distance constraints imposed on $C$. For a code of minimum distance $d$, the constraints are $B_i = 0$ for $1 \leq i \leq d-1$, and the resulting optimal value is denoted by $A_{\mathrm{LP}}(n, d)$. Writing $A(n, d)$ for the largest size of a code over $\mathbb{F}_q$ of length $n$ and minimum distance $d$, we have $A(n, d) \leq A_{\mathrm{LP}}(n, d)$. 

\begin{theorem}[{\cite[Th.~5.3.4]{vL1999}}]\label{thm:lp-ecc}
Let $n$ and $d$ be positive integers and let $q \geq 2$. Then
\[
A(n, d) \;\leq\; \max \Big\{ \sum_{i=0}^{n} B_i \;:\; B_0 = 1,\; B_i = 0 \text{ for } 1 \leq i \leq d-1,
\]
\[
B_i \geq 0 \text{ for } 0 \leq i \leq n,\; \sum_{i=0}^{n} B_i K_j(i) \geq 0 \text{ for } j \in \{0, 1, \ldots, n\} \Big\}.
\]
\end{theorem}

The bounds of Section~\ref{sec:lp} fix the message length and constrain the block length, so we restate $A_{\mathrm{LP}}(n,d)$ as a bound on length. A code encoding $q^k$ messages injectively with length $n$ and minimum distance $d$ has $q^k$ codewords, so its existence forces $A_{\mathrm{LP}}(n,d) \geq q^k$. Every admissible length therefore lies in $\{\, n' : A_{\mathrm{LP}}(n',d) \geq q^k \,\}$, and the least element of that set is a lower bound on the length. Therefore, we have 
\begin{equation}
\label{eq:nlp-ecc}
n \;\geq\; \min \{\, n' \;:\; A_{\mathrm{LP}}(n',d) \geq q^k \,\}.
\end{equation}
 
In Section~\ref{sec:lp} we extend this approach to generalized FCPCs. The distance constraints there depend on the partition blocks of the messages rather than on a single minimum distance, so the distance distribution is refined according to the partitions in which a pair of messages lies in different blocks. 
The bound on the length is then read as in \eqref{eq:nlp-ecc}, with $A_{\mathrm{LP}}(n', d)$ replaced by the optimal value of the corresponding program.


\section{Generalized Function-Correcting Partition Codes}\label{sec:main}

In this section, we introduce generalized FCPCs that protect multiple partitions of the message space against different numbers of errors. We define the distance requirement matrix for these codes and use it to characterize the optimal redundancy. We then state the main results of the paper.

\subsection{The Framework}\label{subsec:framework}

\begin{definition}[Generalized FCPC]\label{def:gfcpc}
Let $\mathcal{P}_1, \ldots, \mathcal{P}_H$ be partitions of $\mathbb{F}_q^k$, and let $d_1, d_2, \ldots, d_H$ be positive integers. A systematic encoding $\mathcal{C} : \mathbb{F}_q^k \to \mathbb{F}_q^{k+r}$ is called a $(\mathcal{P}_h : d_h;\, h \in [H])$-FCPC if for each $h \in [H]$ the following holds: if $\mathcal{P}_h = \{P_h^1, P_h^2, \ldots, P_h^{E_h}\}$, where $E_h = |\mathcal{P}_h|$, and $u \in P_h^i$, $v \in P_h^j$ with $i \neq j$, then
\[
d\big(\mathcal{C}(u),\, \mathcal{C}(v)\big) \geq d_h.
\]
When $d_h = 2t_h + 1$ for all $h \in [H]$, such an encoding is also called a $(\mathcal{P}_h, t_h;\, h \in [H])$-FCPC.
\end{definition}

The optimal redundancy of a $(\mathcal{P}_h : d_h;\, h \in [H])$-FCPC is the smallest integer $r$ for which such an encoding exists. We denote it by
$r_{\mathcal{P}_1, \ldots, \mathcal{P}_H}(k : \bm{d})$ when expressed in terms of the distance requirement vector $\bm{d}=(d_1, \ldots, d_H)$,
and alternatively by
$r_{\mathcal{P}_1, \ldots, \mathcal{P}_H}(k,\, \bm{t})$ when expressed in terms of vector $\bm{t}=(t_1, \ldots, t_H)$, which represents error-correction capabilities,
where $t_h = \lfloor(d_h - 1)/2\rfloor$. These two notations satisfy
\[
r_{\mathcal{P}_1, \ldots, \mathcal{P}_H}(k,\, (t_1, \ldots, t_H)) = r_{\mathcal{P}_1, \ldots, \mathcal{P}_H}(k : (2t_1 + 1, \ldots, 2t_H + 1)).
\]

In the special case $H = 1$, with a single partition $\mathcal{P}$ and distance requirement $d$, the vector $\bm{d}$ reduces to the scalar $d$, and we write $r_{\mathcal{P}}(k : d)$ accordingly.

Throughout the paper we assume without loss of generality that the distance requirements are ordered,
\[
d_1 \leq d_2 \leq \cdots \leq d_H.
\]
This ordering is used in the definition of the distance requirement matrix and in the statements of the bounds.

\begin{remark}\label{rem:generalization}
Definition~\ref{def:gfcpc} unifies and generalizes two existing frameworks.
\begin{enumerate}
    \item \textbf{FCPCs for multiple functions \cite{RRHH20261}.} In \cite{RRHH20261}, multiple functions are protected using a single FCPC defined on the join of their domain partitions with a common error-correction level. Our framework generalizes this by allowing different partitions (equivalently, different functions) to be protected against different numbers of errors.
      \item \textbf{FCCs with data protection \cite{RRHH20251}.} When $H = 2$, $\mathcal{P}_1 = \big\{\{u\} : u \in \mathbb{F}_q^k\big\}$ is the finest partition, and $\mathcal{P}_2$ is the domain partition induced by a function $f : \mathbb{F}_q^k \to S$, then a $(\mathcal{P}_h : d_h;\, h \in [2])$-FCPC is precisely an $(f: \, d_1, d_2)$-FCC as defined in \cite{RRHH20251}.
\end{enumerate}
\end{remark}


\subsection{Distance Requirement Matrix}\label{subsec:drm}

We now introduce the distance requirement matrix (DRM) for generalized FCPCs, which provides a tool for characterizing the optimal redundancy. In the subscript of the DRM, we write $\mathcal{P}$ for the collection $\mathcal{P}_1, \ldots, \mathcal{P}_H$ of all $H$ partitions.

\begin{definition}[Distance requirement matrix]\label{def:drm}
Let $\mathcal{P}_1, \mathcal{P}_2, \ldots, \mathcal{P}_H$ be partitions of
$\mathbb{F}_q^k$ with distance requirements
$d_1 \leq d_2 \leq \cdots \leq d_H$.
Let $u_1, u_2, \ldots, u_M \in \mathbb{F}_q^k$. The
\emph{distance requirement matrix} (DRM)
$\mathcal{D}_{\mathcal{P}}(d_h, h \in [H] : u_1, \ldots, u_M)$
for a $(\mathcal{P}_h : d_h;\, h \in [H])$-FCPC is an $M \times M$
matrix $\mathcal{D}$ with entries
\[
\mathcal{D}_{i,j} =
\begin{cases}
0, & \text{if } u_i \text{ and } u_j \text{ are in the same block of } \ \mathcal{P}_h \text{ for all } h \in [H], \\[2mm]
\max(d_{h'} - d(u_i, u_j),\, 0), & \text{otherwise,}
\end{cases}
\]
where $h' = \max\{h \in [H] : u_i, u_j
\text{ are in different blocks of } \mathcal{P}_h\}$,
for $i, j \in \{1, 2, \ldots, M\}$.
\end{definition}

\begin{remark}
Equivalently, $\mathcal{D}_{\mathcal{P}}(d_h, h \in [H] : u_1, \ldots, u_M)
= \max_{h \in [H]} \mathcal{D}_{\mathcal{P}_h}(d_h : u_1, \ldots, u_M)$,
where the maximum is taken entrywise and $\mathcal{D}_{\mathcal{P}_h}(d_h : u_1, \ldots, u_M)$
is the DRM for the single partition $\mathcal{P}_h$ with distance requirement $d_h$.
\end{remark}

The key idea behind the DRM is the following. For a pair $(u_i, u_j)$ that lies in different blocks of some partition $\mathcal{P}_h$, the strongest distance requirement is determined by the partition with the largest index $h'$ for which $u_i$ and $u_j$ are separated. Since $d_1 \leq d_2 \leq \cdots \leq d_H$, this largest-index partition imposes the strongest constraint. The DRM entry then specifies how much additional parity distance is needed beyond what the Hamming distance between $u_i$ and $u_j$ already provides.

\begin{example}\label{ex:drm}
Consider the space $\mathbb{F}_2^4$, and the following three partitions of $\mathbb{F}_2^4$:
\begin{align*}
\mathcal{P}_1 &= \Big\{\{0000, 1000\},\, \{0100, 0010, 1100\},\, \{0001, 1010, 1001, 0110, 0101, 0011, 1110, 1101\},\\
&\qquad \{1011, 0111\},\, \{1111\}\Big\}, \\[1mm]
\mathcal{P}_2 &= \Big\{\{0000\},\, \{1000, 0100, 0010, 0001\},\, \{1100, 1010, 1001, 0110, 0101, 0011\},\\
&\qquad \{1110, 1101, 1011, 0111\},\, \{1111\}\Big\}, \\[1mm]
\mathcal{P}_3 &= \Big\{\{0000, 1000, 0100, 1100\},\, \{0010, 1010, 0110, 1110\},\, \{0001, 1001, 0101, 1101\},\\
&\qquad \{0011, 1011, 0111, 1111\}\Big\}.
\end{align*}
Let
\[
\textcolor{blue}{d_1 = 3}, \qquad \textcolor{cyan}{d_2 = 5}, \qquad \textcolor{red}{d_3 = 7}.
\]
Order the vectors of $\mathbb{F}_2^4$ as
\[
0000,\, 1000,\, 0100,\, 0010,\, 0001,\, 1100,\, 1010,\, 1001,\, 0110,\, 0101,\, 0011,\, 1110,\, 1101,\, 1011,\, 0111,\, 1111.
\]
For this ordering, the distance requirement matrix $\mathcal{D}_{\mathcal{P}}(d_h, h \in [3] : u_1, \ldots, u_{16})$ is given by

\[
\scriptsize
\setlength{\arraycolsep}{3.5pt}
\renewcommand{\arraystretch}{1.15}
\begin{array}{c|cccccccccccccccc}
 & 0000 & 1000 & 0100 & 0010 & 0001 & 1100 & 1010 & 1001 & 0110 & 0101 & 0011 & 1110 & 1101 & 1011 & 0111 & 1111\\
\hline
0000 &
\textcolor{black}{0} &
\textcolor{cyan}{(4)_2} &
\textcolor{cyan}{(4)_2} &
\textcolor{red}{(6)_3} &
\textcolor{red}{(6)_3} &
\textcolor{cyan}{(3)_2} &
\textcolor{red}{(5)_3} &
\textcolor{red}{(5)_3} &
\textcolor{red}{(5)_3} &
\textcolor{red}{(5)_3} &
\textcolor{red}{(5)_3} &
\textcolor{red}{(4)_3} &
\textcolor{red}{(4)_3} &
\textcolor{red}{(4)_3} &
\textcolor{red}{(4)_3} &
\textcolor{red}{(3)_3}
\\
1000 &
\textcolor{cyan}{(4)_2} &
\textcolor{black}{0} &
\textcolor{blue}{(1)_1} &
\textcolor{red}{(5)_3} &
\textcolor{red}{(5)_3} &
\textcolor{cyan}{(4)_2} &
\textcolor{red}{(6)_3} &
\textcolor{red}{(6)_3} &
\textcolor{red}{(4)_3} &
\textcolor{red}{(4)_3} &
\textcolor{red}{(4)_3} &
\textcolor{red}{(5)_3} &
\textcolor{red}{(5)_3} &
\textcolor{red}{(5)_3} &
\textcolor{red}{(3)_3} &
\textcolor{red}{(4)_3}
\\
0100 &
\textcolor{cyan}{(4)_2} &
\textcolor{blue}{(1)_1} &
\textcolor{black}{0} &
\textcolor{red}{(5)_3} &
\textcolor{red}{(5)_3} &
\textcolor{cyan}{(4)_2} &
\textcolor{red}{(4)_3} &
\textcolor{red}{(4)_3} &
\textcolor{red}{(6)_3} &
\textcolor{red}{(6)_3} &
\textcolor{red}{(4)_3} &
\textcolor{red}{(5)_3} &
\textcolor{red}{(5)_3} &
\textcolor{red}{(3)_3} &
\textcolor{red}{(5)_3} &
\textcolor{red}{(4)_3}
\\
0010 &
\textcolor{red}{(6)_3} &
\textcolor{red}{(5)_3} &
\textcolor{red}{(5)_3} &
\textcolor{black}{0} &
\textcolor{red}{(5)_3} &
\textcolor{red}{(4)_3} &
\textcolor{cyan}{(4)_2} &
\textcolor{red}{(4)_3} &
\textcolor{cyan}{(4)_2} &
\textcolor{red}{(4)_3} &
\textcolor{red}{(6)_3} &
\textcolor{cyan}{(3)_2} &
\textcolor{red}{(3)_3} &
\textcolor{red}{(5)_3} &
\textcolor{red}{(5)_3} &
\textcolor{red}{(4)_3}
\\
0001 &
\textcolor{red}{(6)_3} &
\textcolor{red}{(5)_3} &
\textcolor{red}{(5)_3} &
\textcolor{red}{(5)_3} &
\textcolor{black}{0} &
\textcolor{red}{(4)_3} &
\textcolor{red}{(4)_3} &
\textcolor{cyan}{(4)_2} &
\textcolor{red}{(4)_3} &
\textcolor{cyan}{(4)_2} &
\textcolor{red}{(6)_3} &
\textcolor{red}{(3)_3} &
\textcolor{cyan}{(3)_2} &
\textcolor{red}{(5)_3} &
\textcolor{red}{(5)_3} &
\textcolor{red}{(4)_3}
\\
1100 &
\textcolor{cyan}{(3)_2} &
\textcolor{cyan}{(4)_2} &
\textcolor{cyan}{(4)_2} &
\textcolor{red}{(4)_3} &
\textcolor{red}{(4)_3} &
\textcolor{black}{0} &
\textcolor{red}{(5)_3} &
\textcolor{red}{(5)_3} &
\textcolor{red}{(5)_3} &
\textcolor{red}{(5)_3} &
\textcolor{red}{(3)_3} &
\textcolor{red}{(6)_3} &
\textcolor{red}{(6)_3} &
\textcolor{red}{(4)_3} &
\textcolor{red}{(4)_3} &
\textcolor{red}{(5)_3}
\\
1010 &
\textcolor{red}{(5)_3} &
\textcolor{red}{(6)_3} &
\textcolor{red}{(4)_3} &
\textcolor{cyan}{(4)_2} &
\textcolor{red}{(4)_3} &
\textcolor{red}{(5)_3} &
\textcolor{black}{0} &
\textcolor{red}{(5)_3} &
\textcolor{black}{\textbf{0}} &
\textcolor{red}{(3)_3} &
\textcolor{red}{(5)_3} &
\textcolor{cyan}{(4)_2} &
\textcolor{red}{(4)_3} &
\textcolor{red}{(6)_3} &
\textcolor{red}{(4)_3} &
\textcolor{red}{(5)_3}
\\
1001 &
\textcolor{red}{(5)_3} &
\textcolor{red}{(6)_3} &
\textcolor{red}{(4)_3} &
\textcolor{red}{(4)_3} &
\textcolor{cyan}{(4)_2} &
\textcolor{red}{(5)_3} &
\textcolor{red}{(5)_3} &
\textcolor{black}{0} &
\textcolor{red}{(3)_3} &
\textcolor{black}{\textbf{0}} &
\textcolor{red}{(5)_3} &
\textcolor{red}{(4)_3} &
\textcolor{cyan}{(4)_2} &
\textcolor{red}{(6)_3} &
\textcolor{red}{(4)_3} &
\textcolor{red}{(5)_3}
\\
0110 &
\textcolor{red}{(5)_3} &
\textcolor{red}{(4)_3} &
\textcolor{red}{(6)_3} &
\textcolor{cyan}{(4)_2} &
\textcolor{red}{(4)_3} &
\textcolor{red}{(5)_3} &
\textcolor{black}{\textbf{0}} &
\textcolor{red}{(3)_3} &
\textcolor{black}{0} &
\textcolor{red}{(5)_3} &
\textcolor{red}{(5)_3} &
\textcolor{cyan}{(4)_2} &
\textcolor{red}{(4)_3} &
\textcolor{red}{(4)_3} &
\textcolor{red}{(6)_3} &
\textcolor{red}{(5)_3}
\\
0101 &
\textcolor{red}{(5)_3} &
\textcolor{red}{(4)_3} &
\textcolor{red}{(6)_3} &
\textcolor{red}{(4)_3} &
\textcolor{cyan}{(4)_2} &
\textcolor{red}{(5)_3} &
\textcolor{red}{(3)_3} &
\textcolor{black}{\textbf{0}} &
\textcolor{red}{(5)_3} &
\textcolor{black}{0} &
\textcolor{red}{(5)_3} &
\textcolor{red}{(4)_3} &
\textcolor{cyan}{(4)_2} &
\textcolor{red}{(4)_3} &
\textcolor{red}{(6)_3} &
\textcolor{red}{(5)_3}
\\
0011 &
\textcolor{red}{(5)_3} &
\textcolor{red}{(4)_3} &
\textcolor{red}{(4)_3} &
\textcolor{red}{(6)_3} &
\textcolor{red}{(6)_3} &
\textcolor{red}{(3)_3} &
\textcolor{red}{(5)_3} &
\textcolor{red}{(5)_3} &
\textcolor{red}{(5)_3} &
\textcolor{red}{(5)_3} &
\textcolor{black}{0} &
\textcolor{red}{(4)_3} &
\textcolor{red}{(4)_3} &
\textcolor{cyan}{(4)_2} &
\textcolor{cyan}{(4)_2} &
\textcolor{cyan}{(3)_2}
\\
1110 &
\textcolor{red}{(4)_3} &
\textcolor{red}{(5)_3} &
\textcolor{red}{(5)_3} &
\textcolor{cyan}{(3)_2} &
\textcolor{red}{(3)_3} &
\textcolor{red}{(6)_3} &
\textcolor{cyan}{(4)_2} &
\textcolor{red}{(4)_3} &
\textcolor{cyan}{(4)_2} &
\textcolor{red}{(4)_3} &
\textcolor{red}{(4)_3} &
\textcolor{black}{0} &
\textcolor{red}{(5)_3} &
\textcolor{red}{(5)_3} &
\textcolor{red}{(5)_3} &
\textcolor{red}{(6)_3}
\\
1101 &
\textcolor{red}{(4)_3} &
\textcolor{red}{(5)_3} &
\textcolor{red}{(5)_3} &
\textcolor{red}{(3)_3} &
\textcolor{cyan}{(3)_2} &
\textcolor{red}{(6)_3} &
\textcolor{red}{(4)_3} &
\textcolor{cyan}{(4)_2} &
\textcolor{red}{(4)_3} &
\textcolor{cyan}{(4)_2} &
\textcolor{red}{(4)_3} &
\textcolor{red}{(5)_3} &
\textcolor{black}{0} &
\textcolor{red}{(5)_3} &
\textcolor{red}{(5)_3} &
\textcolor{red}{(6)_3}
\\
1011 &
\textcolor{red}{(4)_3} &
\textcolor{red}{(5)_3} &
\textcolor{red}{(3)_3} &
\textcolor{red}{(5)_3} &
\textcolor{red}{(5)_3} &
\textcolor{red}{(4)_3} &
\textcolor{red}{(6)_3} &
\textcolor{red}{(6)_3} &
\textcolor{red}{(4)_3} &
\textcolor{red}{(4)_3} &
\textcolor{cyan}{(4)_2} &
\textcolor{red}{(5)_3} &
\textcolor{red}{(5)_3} &
\textcolor{black}{0} &
\textcolor{black}{\textbf{0}} &
\textcolor{cyan}{(4)_2}
\\
0111 &
\textcolor{red}{(4)_3} &
\textcolor{red}{(3)_3} &
\textcolor{red}{(5)_3} &
\textcolor{red}{(5)_3} &
\textcolor{red}{(5)_3} &
\textcolor{red}{(4)_3} &
\textcolor{red}{(4)_3} &
\textcolor{red}{(4)_3} &
\textcolor{red}{(6)_3} &
\textcolor{red}{(6)_3} &
\textcolor{cyan}{(4)_2} &
\textcolor{red}{(5)_3} &
\textcolor{red}{(5)_3} &
\textcolor{black}{\textbf{0}} &
\textcolor{black}{0} &
\textcolor{cyan}{(4)_2}
\\
1111 &
\textcolor{red}{(3)_3} &
\textcolor{red}{(4)_3} &
\textcolor{red}{(4)_3} &
\textcolor{red}{(4)_3} &
\textcolor{red}{(4)_3} &
\textcolor{red}{(5)_3} &
\textcolor{red}{(5)_3} &
\textcolor{red}{(5)_3} &
\textcolor{red}{(5)_3} &
\textcolor{red}{(5)_3} &
\textcolor{cyan}{(3)_2} &
\textcolor{red}{(6)_3} &
\textcolor{red}{(6)_3} &
\textcolor{cyan}{(4)_2} &
\textcolor{cyan}{(4)_2} &
\textcolor{black}{0}
\end{array}
\]
\normalsize

Here black entries correspond to pairs that lie in the same block of $\mathcal{P}_1$, $\mathcal{P}_2$, and $\mathcal{P}_3$. Among these, the off-diagonal pairs are shown in bold to distinguish them from the diagonal entries, which are zero because $u_i = u_j$.
 Entries of the form $(\cdot)_1$ (shown in \textcolor{blue}{blue}) correspond to pairs for which
\[
h' = \max\{h \in [3] : u_i, u_j \text{ lie in different blocks of } \mathcal{P}_h\} = 1.
\]
Entries of the form $(\cdot)_2$ (in \textcolor{cyan}{cyan}) correspond to $h' = 2$, and entries of the form $(\cdot)_3$ (in \textcolor{red}{red}) correspond to $h' = 3$.

\end{example}

The DRM provides an exact characterization of the optimal redundancy of a generalized FCPC, by reducing the problem to finding the shortest $\mathcal{D}$-code for the DRM.

\begin{theorem}\label{thm:drm_exact}
The optimal redundancy of a $(\mathcal{P}_h : d_h;\, h \in [H])$-FCPC satisfies
\[
r_{\mathcal{P}_1, \ldots, \mathcal{P}_H}(k : \bm{d}) = N\big(\mathcal{D}_{\mathcal{P}}(d_h, h \in [H] : u_1, \ldots, u_{q^k})\big),
\]
where $\bm{d}=(d_1, \ldots, d_H)$ and $u_1, \ldots, u_{q^k}$ is a fixed ordering of all vectors in $\mathbb{F}_q^k$.
\end{theorem}

\begin{proof}
A systematic encoding $\mathcal{C} : \mathbb{F}_q^k \to \mathbb{F}_q^{k+r}$ maps each message $u_i$ to the codeword $\mathcal{C}(u_i) = (u_i, p_i)$, where $p_i \in \mathbb{F}_q^r$ is the parity vector. For any two messages $u_i$ and $u_j$, the codeword distance satisfies $d(\mathcal{C}(u_i), \mathcal{C}(u_j)) = d(u_i, u_j) + d(p_i, p_j)$.

If $u_i$ and $u_j$ are in the same block of $\mathcal{P}_h$ for all $h \in [H]$, no distance constraint is imposed between them, so the DRM entry is $0$. Otherwise, let $h' = \max\{h \in [H] : u_i, u_j \text{ are in different blocks of } \mathcal{P}_h\}$. Since $d_1 \leq \cdots \leq d_H$, the binding constraint is $d(\mathcal{C}(u_i), \mathcal{C}(u_j)) \geq d_{h'}$, which requires $d(p_i, p_j) \geq \max(d_{h'} - d(u_i, u_j), 0)$.

Therefore, the encoding $\mathcal{C}$ is a valid $(\mathcal{P}_h : d_h;\, h \in [H])$-FCPC if and only if the parity vectors $p_1, \ldots, p_{q^k}$ form a $\mathcal{D}$-code for the DRM $\mathcal{D}_{\mathcal{P}}(d_h, h \in [H] : u_1, \ldots, u_{q^k})$. Consequently, the optimal redundancy equals the minimum length of such a $\mathcal{D}$-code.
\end{proof}


\subsection{Main Results}\label{subsec:summary}

We now state the main results of the paper. Detailed proofs, examples, and the accompanying discussion appear in the sections indicated. Throughout, $\mathcal{P}_1, \ldots, \mathcal{P}_H$ are partitions of $\mathbb{F}_q^k$ with distance requirements $d_1 \leq \cdots \leq d_H$ and $\bm{d} = (d_1, \ldots, d_H)$.

\subsubsection{Combinatorial bounds}
Section~\ref{sec:bounds} derives general bounds on the optimal redundancy. For each $h \in [H]$, let $\mathcal{Q}_h = \mathcal{P}_h \vee \mathcal{P}_{h+1} \vee \cdots \vee \mathcal{P}_H$. Theorem~\ref{thm:lower} and Theorem~\ref{thm:upper} give
\[
\max_{h \in [H]}\, r_{\mathcal{Q}_h}(k : d_h)
\ \leq\ r_{\mathcal{P}_1, \ldots, \mathcal{P}_H}(k : \bm{d}) \ \leq\
\min_{\mathcal{A}} \sum_{A \in \mathcal{A}} r_{\big(\bigvee_{i \in A} \mathcal{P}_i\big)}\!\big(k :\, d_{A}\big),
\]
where the minimum is over all partitions $\mathcal{A}$ of $[H]$ and $d_A = \max_{i \in A} d_i$. Theorem~\ref{thm:drm_lower} applies Theorem~\ref{thm:drm_exact} to submatrices, giving $N\big(\mathcal{D}_{\mathcal{P}}(d_h, h \in [H] : u_1, \ldots, u_M)\big) \leq r_{\mathcal{P}_1, \ldots, \mathcal{P}_H}(k : \bm{d})$ for any $u_1, \ldots, u_M \in \mathbb{F}_q^k$. For $q = 2$ and $H = 2$ this is applied to three-vector submatrices. If $\mathcal{P}_2$ contains vectors $u, v, w$ in three distinct blocks with $d(u,v) = d(u,w) = 1$ and $d(v,w) = 2$, then Proposition~\ref{prop:binary_P2} gives $r_{\mathcal{P}_1, \mathcal{P}_2}(k : (d_1, d_2)) \geq \lceil 3d_2/2 - 2 \rceil$. If instead some block of $\mathcal{P}_2$ contains vectors in different blocks of $\mathcal{P}_1$ satisfying one of the two neighborhood conditions stated there, then Theorem~\ref{thm:binary_bound} gives $r_{\mathcal{P}_1, \mathcal{P}_2}(k : (d_1, d_2)) \geq \lceil d_2 + d_1/2 - 2 \rceil$.

\subsubsection{Linear programming bound}
Section~\ref{sec:lp} derives a linear programming lower bound on the optimal redundancy of GFCPCs over $\mathbb{F}_q$.
For $T \subseteq [H]$, write $\mathcal{P}_T = \bigvee_{j \in T} \mathcal{P}_j$, and let $\widetilde{E}(\mathcal{P}_T)$ denote the effective number of blocks of $\mathcal{P}_T$, defined in Section~\ref{subsec:lp-dist}. Theorem~\ref{thm:lp-multi} gives
\[
r_{\mathcal{P}_1, \ldots, \mathcal{P}_H}(k : \bm{d}) \ge \min\bigl\{\, n : M_{\mathrm{LP}}\bigl(n, d_1, \dots, d_H, \{\widetilde{E}(\mathcal{P}_T)\}_{\emptyset \ne T \subseteq [H]}\bigr) \ge 2^k \,\bigr\} - k,
\]
where $M_{\mathrm{LP}}\bigl(n, d_1, \dots, d_H, \{\widetilde{E}(\mathcal{P}_T)\}_{\emptyset \ne T \subseteq [H]}\bigr)$ is the optimum value of LP\ref{lp:multi}, given in Section~\ref{subsec:lp-multi}.

 The programs for a single partition and for two partitions are stated separately as LP\ref{lp:single} and LP\ref{lp:two}, with the corresponding bounds in Theorem~\ref{thm:lp-single} and Theorem~\ref{thm:lp-two}.
 The case $H = 1$ applies to any FCPC, and hence to any FCC. The case $H = 2$ with $\mathcal{P}_1$ the finest partition applies to any FCC with data protection. No linear programming bound was given for these classes before this work. Several examples are given in Section \ref{sec:lp} to illustrate this bound.

\subsubsection{Multi-step construction}
Section~\ref{sec:construction} presents an $H$-step construction method. Step~1 protects the join $\mathcal{Q}_1$ of all $H$ partitions at distance $d_1$. Step~$h$ appends parity to the image of the previous step and upgrades the distance to $d_h$ for the join $\mathcal{Q}_h$. Theorem~\ref{thm:construction} proves its correctness. 
Let $r^{\mathrm{ms}}_{\mathcal{P}_1, \ldots, \mathcal{P}_H}(k : \bm{d})$ denote the smallest total
redundancy $\sum_{h=1}^{H} r_h$ over all encodings produced by the $H$-step construction, where
$r_h$ is the redundancy added at Step~$h$.  Theorem~\ref{thm:construction} shows that every such
encoding is a valid $(\mathcal{P}_h : d_h;\, h \in [H])$-FCPC. Hence
\[
r_{\mathcal{P}_1, \ldots, \mathcal{P}_H}(k : \bm{d}) \leq r^{\mathrm{ms}}_{\mathcal{P}_1, \ldots, \mathcal{P}_H}(k : \bm{d}).
\]
Examples~\ref{ex:ms-sum} and~\ref{ex:ms-join} show that the construction can use strictly smaller redundancy than the sum of the individual FCPC redundancies and than a single FCPC for the join partition at the largest distance.



\section{Combinatorial Bounds on Optimal Redundancy}\label{sec:bounds}

In this section, we derive general lower and upper bounds on the optimal redundancy of generalized FCPCs. We then use the distance requirement matrix of Section~\ref{subsec:drm} to obtain lower bounds from submatrices. Finally, we present improved lower bounds for two partitions over the binary field under specific structural conditions.

\subsection{General Lower Bound}

\begin{theorem}\label{thm:lower}
Let $\mathcal{P}_1, \mathcal{P}_2, \ldots, \mathcal{P}_H$ be partitions of $\mathbb{F}_q^k$ with distance requirements $d_1 \leq d_2 \leq \cdots \leq d_H$. For each $h \in [H]$, define $\mathcal{Q}_h = \mathcal{P}_h \vee \mathcal{P}_{h+1} \vee \cdots \vee \mathcal{P}_H$. Then the optimal redundancy of a $(\mathcal{P}_h : d_h;\, h \in [H])$-FCPC is lower bounded by
\[
\max_{h \in [H]}\, r_{\mathcal{Q}_h}(k : d_h) \leq r_{\mathcal{P}_1, \ldots, \mathcal{P}_H}(k : \bm{d}),
\]
where $\bm{d}=(d_1, \ldots, d_H)$.
\end{theorem}

\begin{proof}
Let $\mathcal{C} : \mathbb{F}_q^k \to \mathbb{F}_q^{k+r}$ be an optimal $(\mathcal{P}_h : d_h;\, h \in [H])$-FCPC with redundancy $r = r_{\mathcal{P}_1, \ldots, \mathcal{P}_H}(k : \bm{d})$. Fix any $h \in [H]$. We show that $\mathcal{C}$ is also a valid $(\mathcal{Q}_h : d_h)$-FCPC.

Let $u, v \in \mathbb{F}_q^k$ lie in different blocks of $\mathcal{Q}_h = \mathcal{P}_h \vee \mathcal{P}_{h+1} \vee \cdots \vee \mathcal{P}_H$. Since the join $\mathcal{Q}_h$ is a refinement of each of $\mathcal{P}_h, \mathcal{P}_{h+1}, \ldots, \mathcal{P}_H$, the vectors $u$ and $v$ must lie in different blocks of $\mathcal{P}_\ell$ for some $\ell \in \{h, h+1, \ldots, H\}$. Since $\mathcal{C}$ is a $(\mathcal{P}_h : d_h;\, h \in [H])$-FCPC, we have
\[
d\big(\mathcal{C}(u),\, \mathcal{C}(v)\big) \geq d_\ell \geq d_h,
\]
where the last inequality follows from the ordering $d_h \leq d_{h+1} \leq \cdots \leq d_H$. Therefore, $\mathcal{C}$ satisfies the distance requirement of a $(\mathcal{Q}_h : d_h)$-FCPC, which gives
\[
r_{\mathcal{Q}_h}(k : d_h) \leq r.
\]
Since this holds for every $h \in [H]$, taking the maximum over all $h$ yields the desired bound.
\end{proof}

\begin{remark}\label{rem:lower_special}
For $h = H$, we have $\mathcal{Q}_H = \mathcal{P}_H$, and the bound gives $r_{\mathcal{P}_H}(k : d_H) \leq r_{\mathcal{P}_1, \ldots, \mathcal{P}_H}(k : \bm{d})$. For $h = 1$, we have $\mathcal{Q}_1 = \mathcal{P}_1 \vee \cdots \vee \mathcal{P}_H$, and the bound gives $r_{\mathcal{P}_1 \vee \cdots \vee \mathcal{P}_H}(k : d_1) \leq r_{\mathcal{P}_1, \ldots, \mathcal{P}_H}(k : \bm{d})$.
\end{remark}

\subsection{General Upper Bound}

\begin{theorem}\label{thm:upper}
Let $\mathcal{P}_1, \mathcal{P}_2, \ldots, \mathcal{P}_H$ be partitions of $\mathbb{F}_q^k$ with distance requirements $d_1 \leq d_2 \leq \cdots \leq d_H$, and let $\bm{d} = (d_1, \ldots, d_H)$. Then
\[
r_{\mathcal{P}_1, \ldots, \mathcal{P}_H}(k : \bm{d}) \leq \min_{\mathcal{A}} \sum_{A \in \mathcal{A}} r_{\big(\bigvee_{i \in A} \mathcal{P}_i\big)}\!\big(k :\, d_{A}\big),
\]
where the minimum is taken over all partitions $\mathcal{A}$ of the index set $[H]$, and $d_{A} = \max_{i \in A}\, d_i$.
\end{theorem}

\begin{proof}
Let $\mathcal{A} = \{A_1, A_2, \ldots, A_{\alpha}\}$ be a partition of $[H]$ with $|\mathcal{A}| = \alpha$. For each $j \in [\alpha]$, denote
\[
r_j = r_{\big(\bigvee_{i \in A_j} \mathcal{P}_i\big)}\!\big(k :\, d_{A_j}\big),
\]
and let $\mathcal{C}_j$ be an optimal FCPC for the partition $\bigvee_{i \in A_j} \mathcal{P}_i$ with distance requirement $d_{A_j}$ and redundancy $r_j$. For each $j \in [\alpha]$ and message $u \in \mathbb{F}_q^k$, let $c_j(u)$ denote the parity vector of $\mathcal{C}_j$ corresponding to $u$.

Define the encoding
\[
\mathcal{C} : \mathbb{F}_q^k \to \mathbb{F}_q^{k + \sum_{j=1}^{\alpha} r_j}
\]
as
\[
\mathcal{C}(u) = \big(u,\, c_1(u),\, c_2(u),\, \ldots,\, c_{\alpha}(u)\big).
\]

We now show that $\mathcal{C}$ is a valid $(\mathcal{P}_h : d_h;\, h \in [H])$-FCPC. Fix any $h \in [H]$, and let $u, v \in \mathbb{F}_q^k$ lie in different blocks of $\mathcal{P}_h$. Since $\mathcal{A}$ is a partition of $[H]$, there exists $j \in [\alpha]$ such that $h \in A_j$.

Since the join $\bigvee_{i \in A_j} \mathcal{P}_i$ is a refinement of $\mathcal{P}_h$ (as $h \in A_j$), the vectors $u$ and $v$ also lie in different blocks of $\bigvee_{i \in A_j} \mathcal{P}_i$. By the construction of $\mathcal{C}_j$, we have
\[
d\big((u,\, c_j(u)),\, (v,\, c_j(v))\big) \geq d_{A_j}.
\]
Therefore,
\[
d\big(\mathcal{C}(u),\, \mathcal{C}(v)\big) \geq d\big((u,\, c_j(u)),\, (v,\, c_j(v))\big) \geq d_{A_j} \geq d_h,
\]
where the last inequality follows from $h \in A_j$ and the definition $d_{A_j} = \max_{i \in A_j} d_i \geq d_h$.

Thus, $\mathcal{C}$ satisfies the distance requirement for every partition $\mathcal{P}_h$, and hence is a valid $(\mathcal{P}_h : d_h;\, h \in [H])$-FCPC with total redundancy $\sum_{j=1}^{\alpha} r_j$. Since $\mathcal{A}$ was an arbitrary partition of $[H]$, taking the minimum over all such partitions completes the proof.
\end{proof}

The following example illustrates that the minimum in the upper bound of Theorem~\ref{thm:upper} is not always achieved by one of the two extreme partitions of $[H]$ (the finest or the coarsest). In particular, it demonstrates that a non-trivial grouping of the partitions can yield a strictly tighter upper bound.

\begin{example}\label{ex:upper-grouping}
Let $k = 3$ and $q = 2$. For $1 \leq i \leq 3$, consider the coordinate projection functions
\[
f_i : \mathbb{F}_2^3 \to \mathbb{F}_2, \qquad f_i(u_1, u_2, u_3) = u_i.
\]
The domain partitions induced by these functions are
\begin{align*}
\mathcal{P}_1 &= \big\{\{000, 001, 010, 011\},\, \{100, 101, 110, 111\}\big\}, \\
\mathcal{P}_2 &= \big\{\{000, 001, 100, 101\},\, \{010, 011, 110, 111\}\big\}, \\
\mathcal{P}_3 &= \big\{\{000, 010, 100, 110\},\, \{001, 011, 101, 111\}\big\}.
\end{align*}

We set the distance requirements $d_1 = 3$, $d_2 = 3$, and $d_3 = 11$ (corresponding to error-correction capabilities $t_1 = 1$, $t_2 = 1$, and $t_3 = 5$, respectively). We compute the relevant pairwise joins:
\begin{align*}
\mathcal{P}_1 \vee \mathcal{P}_2 &= \big\{\{000, 001\},\, \{010, 011\},\, \{100, 101\},\, \{110, 111\}\big\}, \\
\mathcal{P}_1 \vee \mathcal{P}_3 &= \big\{\{000, 010\},\, \{001, 011\},\, \{100, 110\},\, \{101, 111\}\big\}, \\
\mathcal{P}_2 \vee \mathcal{P}_3 &= \big\{\{000, 100\},\, \{001, 101\},\, \{010, 110\},\, \{011, 111\}\big\}, \\
\mathcal{P}_1 \vee \mathcal{P}_2 \vee \mathcal{P}_3 &= \big\{\{000\},\, \{001\},\, \{010\},\, \{011\},\, \{100\},\, \{101\},\, \{110\},\, \{111\}\big\}.
\end{align*}

Table~\ref{tab:ex3_encodings} shows optimal encodings for each individual and join partition at the required distance. These encodings establish the redundancy values used in the evaluation of the Theorem~\ref{thm:upper} upper bound. The optimality of each redundancy value in Table~\ref{tab:ex3_encodings} was verified by an exhaustive computer search.

\begin{table*}[ht]
\centering
\caption{Optimal encodings for individual and join partitions in Example~\ref{ex:upper-grouping}.}
\label{tab:ex3_encodings}

\begin{subtable}[t]{0.25\textwidth}
\centering
\caption{$\mathcal{P}_1$, $d_1\!=\!3$, $r\!=\!2$.}
\label{tab:ex3_P1}
\begin{tabular}{c|c}
\hline
Block & Red. \\
\hline
$\{000,001,010,011\}$ & $00$ \\
$\{100,101,110,111\}$ & $11$ \\
\hline
\end{tabular}
\end{subtable}
\hfill
\begin{subtable}[t]{0.25\textwidth}
\centering
\caption{$\mathcal{P}_2$, $d_2\!=\!3$, $r\!=\!2$.}
\label{tab:ex3_P2}
\begin{tabular}{c|c}
\hline
Block & Red. \\
\hline
$\{000,001,100,101\}$ & $00$ \\
$\{010,011,110,111\}$ & $11$ \\
\hline
\end{tabular}
\end{subtable}
\hfill
\begin{subtable}[t]{0.40\textwidth}
\centering
\caption{$\mathcal{P}_3$, $d_3\!=\!11$, $r\!=\!10$.}
\label{tab:ex3_P3}
\begin{tabular}{c|c}
\hline
Block & Redundancy \\
\hline
$\{000,010,100,110\}$ & $000\;000\;000\;0$ \\
$\{001,011,101,111\}$ & $111\;111\;111\;1$ \\
\hline
\end{tabular}
\end{subtable}

\vspace{1em}

\begin{subtable}[t]{0.22\textwidth}
\centering
\caption{$\mathcal{P}_1\!\vee\!\mathcal{P}_2$, $d_2\!=\!3$, $r\!=\!3$.}
\label{tab:ex3_P1vP2}
\begin{tabular}{c|c}
\hline
Block & Redundancy \\
\hline
$\{000,001\}$ & $000$ \\
$\{010,011\}$ & $011$ \\
$\{110,111\}$ & $100$ \\
$\{100,101\}$ & $111$ \\
\hline
\end{tabular}
\end{subtable}
\hfill
\begin{subtable}[t]{0.36\textwidth}
\centering
\caption{$\mathcal{P}_2\!\vee\!\mathcal{P}_3$, $d_3\!=\!11$, $r\!=\!15$.}
\label{tab:ex3_P2vP3}
\begin{tabular}{c|c}
\hline
Block & Redundancy\\
\hline
$\{000,100\}$ & $0 0 0\; 0 0 0\; 0 0 0\; 0 0 0\; 0 0 0$ \\
$\{001,101\}$ & $0 0 0\; 0 0 1\; 1 1 1\; 1 1 1 \;1 1 1$ \\
$\{111,011\}$ & $1 1 1\; 1 1 0\; 0 0 0\; 0 0 1 \;1 1 1$ \\
$\{110,010\}$ & $1 1 1 \;1 1 1\; 1 1 1 \;1 1 0 \;0 0 0$ \\
\hline
\end{tabular}
\end{subtable}
\hfill
\begin{subtable}[t]{0.36\textwidth}
\centering
\caption{$\mathcal{P}_1\!\vee\!\mathcal{P}_3$, $d_3\!=\!11$, $r\!=\!15$.}
\label{tab:ex3_P1vP3}
\begin{tabular}{c|c}
\hline
Block & Redundancy \\
\hline
$\{000,010\}$ & $0 0 0\; 0 0 0\; 0 0 0\; 0 0 0\; 0 0 0$ \\
$\{001,011\}$ & $0 0 0\; 0 0 1\; 1 1 1 \;1 1 1\; 1 1 1$ \\
$\{111,101\}$ & $1 1 1\; 1 1 0\; 0 0 0 \;0 1 1 \;1 1 1$ \\
$\{100,110\}$ & $1 1 1 \;1 1 1 \;1 1 1 \;1 0 0\; 0 0 0$ \\
\hline
\end{tabular}
\end{subtable}

\vspace{1em}

\begin{subtable}[t]{\textwidth}
\centering
\caption{$\mathcal{P}_1\!\vee\!\mathcal{P}_2\!\vee\!\mathcal{P}_3$, $d_3\!=\!11$, $r\!=\!17$.}
\label{tab:ex3_P1vP2vP3}
\begin{tabular}{c|c}
\hline
Block & Redundancy\\
\hline
$\{000\}$ & $0 0 0\;0 0 0\; 0 0 0\;0 0 0\;000 \;00$ \\
$\{100\}$ & $0 0 0 \;000\;0 11\; 1 1 1 \;1 1 1\;11$ \\
$\{001\}$ & $0 0 1 \;1 1 1\; 1 0 0 \;0 0 0 \;111\;11$ \\
$\{010\}$ & $0 0 1\; 1 1 1\; 1 1 1\; 1 1 1\; 0 0 0\; 0 0$ \\
$\{110\}$ & $1 1 0\; 0 1 1\; 1 0 0\; 0 1 1\; 0 0 1\; 1 1$ \\
$\{101\}$ & $1 1 0 \;0 1 1\; 1 1 1\; 1 0 0\; 1 1 0\; 0 0$ \\
$\{011\}$ & $1 1 1\; 1 0 0 \;0 0 0\; 1 1 1\; 1 1 0\; 0 1$ \\
$\{111\}$ & $1 1 1\; 1 0 0\; 0 1 1\; 0 0 0\; 0 0 1\; 1 0$ \\
\hline
\end{tabular}
\end{subtable}

\end{table*}

We now evaluate the upper bound of Theorem~\ref{thm:upper} for each partition $\mathcal{A}$ of the index set $[3] = \{1, 2, 3\}$. The set $[3]$ admits five partitions, and the corresponding bound values are listed in Table~\ref{tab:ex3_theorem2}.

\begin{table}[ht]
\centering
\caption{Evaluation of the Theorem~\ref{thm:upper} upper bound for all partitions of $[3]$.}
\label{tab:ex3_theorem2}
\begin{tabular}{c|c|c}
\hline
Partition $\mathcal{A}$ of $[3]$ & Expression & Value \\
\hline
$\big\{\{1\}, \{2\}, \{3\}\big\}$ & $r_{\mathcal{P}_1}(k : d_1) + r_{\mathcal{P}_2}(k : d_2) + r_{\mathcal{P}_3}(k : d_3)$ & $2 + 2 + 10 = 14$ \\
$\big\{\{1, 2\}, \{3\}\big\}$ & $r_{\mathcal{P}_1 \vee \mathcal{P}_2}(k : d_2) + r_{\mathcal{P}_3}(k : d_3)$ & $3 + 10 = \mathbf{13}$ \\
$\big\{\{1, 3\}, \{2\}\big\}$ & $r_{\mathcal{P}_1 \vee \mathcal{P}_3}(k : d_3) + r_{\mathcal{P}_2}(k : d_2)$ & $15 + 2 = 17$ \\
$\big\{\{2, 3\}, \{1\}\big\}$ & $r_{\mathcal{P}_2 \vee \mathcal{P}_3}(k : d_3) + r_{\mathcal{P}_1}(k : d_1)$ & $15 + 2 = 17$ \\
$\big\{\{1, 2, 3\}\big\}$ & $r_{\mathcal{P}_1 \vee \mathcal{P}_2 \vee \mathcal{P}_3}(k : d_3)$ & $17$ \\
\hline
\end{tabular}
\end{table}

The minimum over all partitions of $[3]$ is $13$, achieved by the partition $\mathcal{A} = \big\{\{1, 2\}, \{3\}\big\}$. This is strictly less than the value $14$ given by the finest partition and the value $17$ given by the coarsest partition. The partition $\mathcal{A} = \big\{\{1, 2\}, \{3\}\big\}$ is effective because it groups $\mathcal{P}_1$ and $\mathcal{P}_2$ together, which share the same low distance requirement $d_1 = d_2 = 3$, so the joint redundancy $r_{\mathcal{P}_1 \vee \mathcal{P}_2}(k : 3) = 3$ is only slightly larger than the individual values of $2$. In contrast, grouping either $\mathcal{P}_1$ or $\mathcal{P}_2$ with $\mathcal{P}_3$ forces the entire group to be protected at the much larger distance $d_3 = 11$, causing a significant increase in redundancy.

By Theorem~\ref{thm:upper}, we obtain
\[
r_{\mathcal{P}_1, \mathcal{P}_2, \mathcal{P}_3}(k : (d_1, d_2, d_3)) \leq 13.
\]

Furthermore, applying the multi-step construction method from Section~\ref{sec:construction}, we construct a valid $(\mathcal{P}_h : d_h;\, h \in [3])$-FCPC as follows. In Step~1, we construct a $(\mathcal{P}_1 \vee \mathcal{P}_2 \vee \mathcal{P}_3)$-encoding with distance $d_1 = 3$ and redundancy $r_1 = 3$. Since $d_1 = d_2 = 3$, Step~2 requires no additional redundancy ($r_2 = 0$). In Step~3, we upgrade the distance for $\mathcal{P}_3$ from $d_2 = 3$ to $d_3 = 11$ by appending redundancy $r_3 = 8$. Table~\ref{tab:ex3_multistep} shows the resulting encoding.

\begin{table}[ht]
\centering
\caption{Multi-step construction for Example~\ref{ex:upper-grouping}.}
\label{tab:ex3_multistep}
\begin{tabular}{c|c|c|c}
\hline
Message $u$ & \makecell{Step~1 red. \\ $\mathcal{P}_1\!\vee\!\mathcal{P}_2\!\vee\!\mathcal{P}_3$, $d_1\!=\!3$} & \makecell{Step~2 red. \\ $\mathcal{P}_2\!\vee\!\mathcal{P}_3$, $d_2\!=\!3$} & \makecell{Step~3 red. \\ $\mathcal{P}_3$, $d_3\!=\!11$} \\
\hline
$000$ & $000$ & $-$ & $000\;000\;00$ \\
$100$ & $110$ & $-$ & $000\;000\;00$ \\
$001$ & $101$ & $-$ & $111\;111\;11$ \\
$010$ & $011$ & $-$ & $000\;000\;00$ \\
$110$ & $101$ & $-$ & $000\;000\;00$ \\
$101$ & $011$ & $-$ & $111\;111\;11$ \\
$011$ & $110$ & $-$ & $111\;111\;11$ \\
$111$ & $000$ & $-$ & $111\;111\;11$ \\
\hline
\end{tabular}
\end{table}

The multi-step construction achieves total redundancy
\[
r^{\mathrm{ms}}_{\mathcal{P}_1, \mathcal{P}_2, \mathcal{P}_3}(k :\, (3, 3, 11)) = 3 + 0 + 8 = 11.
\]

We now evaluate the lower bound from Theorem~\ref{thm:lower}. Recall that $\mathcal{Q}_h = \mathcal{P}_h \vee \mathcal{P}_{h+1} \vee \cdots \vee \mathcal{P}_3$ for $h \in [3]$. We have
\begin{align*}
r_{\mathcal{Q}_1}(k : d_1) &= r_{\mathcal{P}_1 \vee \mathcal{P}_2 \vee \mathcal{P}_3}(k : 3) = 3, \\
r_{\mathcal{Q}_2}(k : d_2) &= r_{\mathcal{P}_2 \vee \mathcal{P}_3}(k : 3) = 3, \\
r_{\mathcal{Q}_3}(k : d_3) &= r_{\mathcal{P}_3}(k : 11) = 10.
\end{align*}
Therefore, the Theorem~\ref{thm:lower} lower bound gives
\[
\max_{h \in [3]}\, r_{\mathcal{Q}_h}(k : d_h) = \max(3,\, 3,\, 10) = 10.
\]

Combining the lower bound, the multi-step construction, and the upper bounds, we obtain
\[
10 \leq r_{\mathcal{P}_1, \mathcal{P}_2, \mathcal{P}_3}(k :\, (3, 3, 11)) \leq r^{\mathrm{ms}}_{\mathcal{P}_1, \mathcal{P}_2, \mathcal{P}_3}(k :\, (3, 3, 11)) = 11 < 13 < 14 < 17,
\]
where $10$ is the lower bound of  Theorem~\ref{thm:lower}, $13$ is the tightest value of the upper bound of Theorem~\ref{thm:upper}, $14$ is the sum of individual redundancies (the finest partition), and $17$ is the redundancy of a single FCPC for the joint partition at the highest distance (the coarsest partition). The gap between the lower bound and the multi-step construction method is only $1$.
\end{example}


\subsection{Lower Bounds from the Distance Requirement Matrix}\label{subsec:drm-lower}

Theorem~\ref{thm:drm_exact} characterizes the optimal redundancy as the minimum length of a $\mathcal{D}$-code for the full DRM. Computing $N(\mathcal{D}_{\mathcal{P}})$ for the full DRM of size $q^k \times q^k$ can be prohibitively expensive. However, restricting the DRM to a suitably chosen subset of messages yields a computable lower bound on the optimal redundancy.

\begin{theorem}\label{thm:drm_lower}
For any $m$ vectors $u_1, u_2, \ldots, u_m$ in $\mathbb{F}_q^k$, the optimal redundancy of a $(\mathcal{P}_h : d_h;\, h \in [H])$-FCPC is lower bounded by
\[
N\big(\mathcal{D}_{\mathcal{P}}(d_h, h \in [H] : u_1, \ldots, u_m)\big) \leq r_{\mathcal{P}_1, \ldots, \mathcal{P}_H}(k : \bm{d}),
\]
where $\bm{d}=(d_1, \ldots, d_H)$.
Additionally, for $|\mathcal{P}_H| \geq 2$, we have $r_{\mathcal{P}_1, \ldots, \mathcal{P}_H}(k : \bm{d}) \geq d_H - 1$.
\end{theorem}
\begin{proof}
Let $\mathcal{C}$ be an optimal $(\mathcal{P}_h : d_h;\, h \in [H])$-FCPC with parity vectors $p_1, \ldots, p_{q^k}$. By Theorem~\ref{thm:drm_exact}, these parity vectors form a $\mathcal{D}$-code for the full DRM $\mathcal{D}_{\mathcal{P}}(d_h, h \in [H] : u_1, \ldots, u_{q^k})$. Restricting to the subset $\{u_1, \ldots, u_m\}$, the corresponding parity vectors $p_1, \ldots, p_m$ form a $\mathcal{D}$-code for the submatrix $\mathcal{D}_{\mathcal{P}}(d_h, h \in [H] : u_1, \ldots, u_m)$. Since $N(\mathcal{D})$ denotes the minimum length of a $\mathcal{D}$-code, the lower bound follows.

For the second statement, since $|\mathcal{P}_H| \geq 2$, there exist vectors of $\mathbb{F}_q^k$ lying in different blocks of $\mathcal{P}_H$. Any two such vectors are joined by a path in which consecutive vectors are at Hamming distance $1$. Some consecutive pair on this path therefore lies in different blocks of $\mathcal{P}_H$. Fix such a pair $u, v$, so that $d(u, v) = 1$. The DRM for these two vectors has the single off-diagonal entry $\max(d_H - 1, 0) = d_H - 1$. A $\mathcal{D}$-code for this matrix requires parity vectors at distance at least $d_H - 1$, so $r \geq d_H - 1$.
\end{proof}

\begin{remark}\label{rem:plotkin}
The Plotkin bound for $\mathcal{D}$-codes applies directly to generalized FCPCs. By
Theorem~\ref{thm:drm_exact} the optimal redundancy is the minimum length of a $\mathcal{D}$-code for
the DRM, so \cite[Lemma~13]{RRHH20251} gives
\[
r_{\mathcal{P}_1, \ldots, \mathcal{P}_H}(k : \bm{d}) \geq
\frac{2q}{M^2(q-1) - a(q-a)} \sum_{1 \leq i < j \leq M} \mathcal{D}_{i,j},
\qquad a = M \bmod q,
\]
for $\mathcal{D} = \mathcal{D}_{\mathcal{P}}(d_h, h \in [H] : u_1, \ldots, u_M)$ with $M = q^k$. By
Theorem~\ref{thm:drm_lower}, the same bound holds for the DRM of any subset of $M$ messages. For $q = 2$ the bound reduces to the form of
\cite[Lemma~11]{RRHH20251}, stated there for $\mathcal{D}$-codes over the binary field and
attributed to \cite{LBZY2023}. The examples of Section~\ref{sec:lp} evaluate it on the full DRM over
$\mathbb{F}_q$, which means $M = q^k$.
\end{remark}


\subsection{Improved Lower Bounds for the Binary Case}\label{subsec:binary}

In this subsection, we consider the binary case $q = 2$ with two partitions $\mathcal{P}_1$ and $\mathcal{P}_2$ of $\mathbb{F}_2^k$ with distance requirements $d_1 \leq d_2$. We establish improved lower bounds on the optimal redundancy $r_{\mathcal{P}_1, \mathcal{P}_2}(k : (d_1, d_2))$ under specific structural conditions on the partitions. For a partition $\mathcal{P}$ of $\mathbb{F}_2^k$, we write $\mathcal{P}(u) = \mathcal{P}(v)$ to mean that $u$ and $v$ belong to the same block of $\mathcal{P}$, and $\mathcal{P}(u) \neq \mathcal{P}(v)$ to mean that they belong to different blocks.

We begin with a bound that depends only on the partition $\mathcal{P}_2$.

\begin{proposition}\label{prop:binary_P2}
If there exist vectors $u, v, w \in \mathbb{F}_2^k$ that belong to different blocks of $\mathcal{P}_2$ and satisfy $d(u, v) = d(u, w) = 1$ and $d(v, w) = 2$, then
\[
r_{\mathcal{P}_1, \mathcal{P}_2}(k : (d_1, d_2)) \geq \left\lceil \frac{3d_2}{2} - 2 \right\rceil.
\]
\end{proposition}

\begin{proof}
Since $u$, $v$, and $w$ belong to different blocks of $\mathcal{P}_2$, and $d_2 \geq d_1$, the DRM for these three vectors is
\[
\mathcal{D}_{\mathcal{P}}(d_1, d_2 : u, v, w) = \begin{pmatrix} 0 & d_2 - 1 & d_2 - 1 \\ d_2 - 1 & 0 & d_2 - 2 \\ d_2 - 1 & d_2 - 2 & 0 \end{pmatrix}.
\]
Let $\{p_1, p_2, p_3\}$ be a $\mathcal{D}$-code for this DRM with length $r$. For any three binary vectors of length $r$, the sum of pairwise distances satisfies $d(p_1, p_2) + d(p_1, p_3) + d(p_2, p_3) \leq 2r$. Therefore,
\[
(d_2 - 1) + (d_2 - 1) + (d_2 - 2) \leq d(p_1, p_2) + d(p_1, p_3) + d(p_2, p_3) \leq 2r,
\]
which gives $3d_2 - 4 \leq 2r$, yielding the desired bound.
\end{proof}

We now give an example of partitions satisfying the condition of Proposition~\ref{prop:binary_P2} and evaluate the resulting bound.

\begin{example}\label{ex:prop1}
Let $k = 3$ and $q = 2$. Consider the partitions
\begin{align*}
\mathcal{P}_1 &= \big\{\{000, 001, 110, 111\},\; \{010, 011, 100, 101\}\big\}, \\
\mathcal{P}_2 &= \big\{\{000, 011\},\; \{001, 010\},\; \{100, 111\},\; \{101, 110\}\big\}.
\end{align*}
Let $d_1 = 3$ and $d_2 = 5$. The vectors $\mathbf{u} = 000$, $\mathbf{v} = 100$, and $\mathbf{w} = 010$ belong to three distinct blocks of $\mathcal{P}_2$, with $d(\mathbf{u}, \mathbf{v}) = d(\mathbf{u}, \mathbf{w}) = 1$ and $d(\mathbf{v}, \mathbf{w}) = 2$. Therefore, the conditions of Proposition~\ref{prop:binary_P2} are satisfied, and we obtain
\[
r_{\mathcal{P}_1, \mathcal{P}_2}(k : (3, 5)) \geq \left\lceil \frac{3 \cdot 5}{2} - 2 \right\rceil = 6.
\]
\end{example}

The bound in Proposition~\ref{prop:binary_P2} depends only on the partition $\mathcal{P}_2$ and its distance requirement $d_2$, without involving $\mathcal{P}_1$ or $d_1$. The following theorem strengthens this by exploiting the interaction between the two partitions, resulting in a bound that involves both $d_1$ and $d_2$. However, since $d_1 \leq d_2$, the bound of Proposition~\ref{prop:binary_P2} is tighter when the conditions of both results are met.

\begin{theorem}\label{thm:binary_bound}
Let $\mathcal{P}_1$ and $\mathcal{P}_2$ be two partitions of $\mathbb{F}_2^k$ with $d_1 \leq d_2$. Suppose some block $B$ of $\mathcal{P}_2$ contains vectors $v, w$ with $\mathcal{P}_1(v) \neq \mathcal{P}_1(w)$ satisfying one of the following conditions:
\begin{enumerate}
    \item $d(v, w) = 1$ and at least one of $v, w$ has a neighbor outside $B$, or
    \item $d(v, w) = 2$ and $v, w$ have a common neighbor outside $B$.
\end{enumerate}
Then the optimal redundancy of a $(\mathcal{P}_h : d_h;\, h \in [2])$-FCPC is lower bounded by
\[
r_{\mathcal{P}_1, \mathcal{P}_2}(k : (d_1, d_2)) \geq \left\lceil d_2 + \frac{d_1}{2} - 2 \right\rceil.
\]
\end{theorem}

\begin{proof}
We consider each condition separately.

\textbf{Condition~1.} We have $v, w \in B$ with $d(v, w) = 1$ and $\mathcal{P}_1(v) \neq \mathcal{P}_1(w)$. Write $w = v + e_i$ for some $i$. Since at least one of $v, w$ has a neighbor outside $B$, without loss of generality let $u = v + e_j$ with $j \neq i$ be a vector not in $B$. Then $d(u, v) = 1$, $d(v, w) = 1$, and $d(u, w) = 2$. Since $u \notin B$ while $v, w \in B$, the vectors $u$ and $v$ lie in different blocks of $\mathcal{P}_2$, as do $u$ and $w$. Meanwhile, $v$ and $w$ lie in the same block of $\mathcal{P}_2$ but in different blocks of $\mathcal{P}_1$. The DRM for $u, v, w$ is therefore
\[
\mathcal{D}_{\mathcal{P}}(d_1, d_2 : u, v, w) = \begin{pmatrix} 0 & d_2 - 1 & d_2 - 2 \\ d_2 - 1 & 0 & d_1 - 1 \\ d_2 - 2 & d_1 - 1 & 0 \end{pmatrix}.
\]
Let $\{p_1, p_2, p_3\}$ be a $\mathcal{D}$-code for this DRM with length $r$. Then
\[
(d_2 - 1) + (d_2 - 2) + (d_1 - 1) \leq d(p_1, p_2) + d(p_1, p_3) + d(p_2, p_3) \leq 2r,
\]
which gives $2d_2 + d_1 - 4 \leq 2r$, yielding the desired bound.

\textbf{Condition~2.} We have $v, w \in B$ with $d(v, w) = 2$ and $\mathcal{P}_1(v) \neq \mathcal{P}_1(w)$. Let $u$ be a common neighbor of $v$ and $w$ that does not belong to $B$, so $d(u, v) = d(u, w) = 1$. Then $u$ lies in a different block of $\mathcal{P}_2$ from both $v$ and $w$, while $v$ and $w$ are in the same block of $\mathcal{P}_2$ but different blocks of $\mathcal{P}_1$. The DRM for $u, v, w$ is
\[
\mathcal{D}_{\mathcal{P}}(d_1, d_2 : u, v, w) = \begin{pmatrix} 0 & d_2 - 1 & d_2 - 1 \\ d_2 - 1 & 0 & d_1 - 2 \\ d_2 - 1 & d_1 - 2 & 0 \end{pmatrix}.
\]
Let $\{p_1, p_2, p_3\}$ be a $\mathcal{D}$-code for this DRM with length $r$. Then
\[
(d_2 - 1) + (d_2 - 1) + (d_1 - 2) \leq d(p_1, p_2) + d(p_1, p_3) + d(p_2, p_3) \leq 2r,
\]
which again gives $2d_2 + d_1 - 4 \leq 2r$, yielding the desired bound.
\end{proof}

We now present an example illustrating the types of partitions for which the bound in Theorem~\ref{thm:binary_bound} applies and show that the two sufficient conditions are independent of each other.

\begin{example}\label{ex:conditions}
Table~\ref{tab:ex-conditions} exhibits three cases showing that Conditions~1 and~2 of Theorem~\ref{thm:binary_bound} are logically independent. In each case, the vectors $v, w$ that satisfy a condition are underlined in $\mathcal{P}_2$, where they lie in a common block $B$, and the corresponding neighbor $u$ outside $B$ is listed separately.

\begin{table}[h]
\centering
\caption{Cases showing the independence of Conditions~1 and~2 of Theorem~\ref{thm:binary_bound}.}
\label{tab:ex-conditions}
\small
\begin{tabular}{ccllcc}
\hline
Case & Space & \multicolumn{1}{c}{Partitions} & \multicolumn{1}{c}{$(v, w; u)$} & C1 & C2 \\
\hline
(a) & $\mathbb{F}_2^3$ &
\makecell[l]{$\mathcal{P}_1 = \big\{\{000, 001, 100, 101\},\, \{010, 011, 110, 111\}\big\}$ \\
$\mathcal{P}_2 = \big\{\{\underline{000}, 001, \underline{010}, 011\},\, \{100, 101, 110, 111\}\big\}$} &
$(000, 010;\, 100)$ & \checkmark & $\times$ \\[2ex]
\hline 
(b) & $\mathbb{F}_2^2$ &
\makecell[l]{$\mathcal{P}_1 = \big\{\{00, 01\},\, \{10, 11\}\big\}$ \\
$\mathcal{P}_2 = \big\{\{\underline{00}, \underline{11}\},\, \{01, 10\}\big\}$} &
$(00, 11;\, 01)$ & $\times$ & \checkmark \\[2ex]
\hline
(c) & $\mathbb{F}_2^4$ &
\makecell[l]{$\mathcal{P}_1 = \big\{\{0000, 0001, 0010, 0011\},\, \{0100, 0101, 0110, 0111\},$ \\
$\qquad\quad \{1000, 1001, 1010, 1011\},\, \{1100, 1101, 1110, 1111\}\big\}$ \\
$\mathcal{P}_2 = \big\{\{\underline{0000}, \underline{0001}, \underline{0100}, 0110\},\, \{0010, 0011, 0101, 0111\},$ \\
$\qquad\quad \{1000, 1001, 1100, 1110\},\, \{1010, 1011, 1101, 1111\}\big\}$} &
\makecell[l]{$(0000, 0100;\, 1000)$ \\ $(0001, 0100;\, 0101)$} & \checkmark & \checkmark \\
\hline
\end{tabular}
\end{table}

In Case~(a), each pair of vectors at distance $2$ in a block of $\mathcal{P}_2$ has both of its
common neighbors in that block, so Condition~2 fails. In Case~(b), the two vectors of each block of $\mathcal{P}_2$ are at distance
$2$, so Condition~1 fails.

In each of the three cases, for any $d_1 \leq d_2$,
\[
r_{\mathcal{P}_1, \mathcal{P}_2}(k : (d_1, d_2)) \geq \left\lceil d_2 + \frac{d_1}{2} - 2 \right\rceil.
\]
\end{example}

We conclude by noting that the bound in Theorem~\ref{thm:binary_bound} relies on the property that for any three binary vectors of length $r$, the sum of pairwise Hamming distances is at most $2r$. This property is specific to $q = 2$, as the following example shows.

\begin{example}[The bound does not extend to $q \geq 3$]\label{ex:q3}
Let $q = 3$, $k = 2$, $d_1 = 3$, and $d_2 = 5$. Consider the function $f : \mathbb{F}_3^2 \to \mathbb{F}_3$ defined by $f(a, b) = a$, so that $|\mathrm{Im}(f)| = 3$. Taking $\mathcal{P}_1$ to be the finest partition and $\mathcal{P}_2$ to be the domain partition induced by $f$, the bound from Theorem~\ref{thm:binary_bound} (if it were to hold over $\mathbb{F}_3$) would give $r \geq d_2 + \lceil d_1/2 \rceil - 2 = 5 + 2 - 2 = 5$. However, the following encoding achieves redundancy $r = 4$:

\begin{table}[ht]
\centering
\caption{$(\mathcal{P}_h : d_h;\, h \in [2])$-FCPC over $\mathbb{F}_3^2$ with $d_1 = 3$, $d_2 = 5$, and redundancy $4$.}
\label{tab:ex_q3}
\begin{tabular}{c|c|c|c}
\hline
$u$ & $f(u)$ & Parity & Codeword \\
\hline
$00$ & $0$ & $0000$ & $000000$ \\
$01$ & $0$ & $0120$ & $010120$ \\
$02$ & $0$ & $0210$ & $020210$ \\
$10$ & $1$ & $1111$ & $101111$ \\
$11$ & $1$ & $1201$ & $111201$ \\
$12$ & $1$ & $1021$ & $121021$ \\
$20$ & $2$ & $2222$ & $202222$ \\
$21$ & $2$ & $2012$ & $212012$ \\
$22$ & $2$ & $2102$ & $222102$ \\
\hline
\end{tabular}
\end{table}

This encoding achieves minimum distance $d_1 = 3$ between all distinct codewords and minimum distance $d_2 = 5$ between codewords corresponding to different function values, using only $4$ symbols of redundancy. This shows that the binary bound from Theorem~\ref{thm:binary_bound} does not extend to $q \geq 3$.
\end{example}

\begin{remark}
The three-vector argument still applies for $q \geq 3$, but it does not improve on
Theorem~\ref{thm:drm_lower}. Replacing the bound $2r$ by $3r$ in the proof of
Theorem~\ref{thm:binary_bound} gives, under either condition,
\[
r_{\mathcal{P}_1, \mathcal{P}_2}(k : (d_1, d_2)) \geq \left\lceil \frac{2d_2 + d_1 - 4}{3} \right\rceil,
\]
and the same computation in Proposition~\ref{prop:binary_P2} gives
$r_{\mathcal{P}_1, \mathcal{P}_2}(k : (d_1, d_2)) \geq d_2 - 1$. Since $d_1 \leq d_2$, we have
$2d_2 + d_1 - 4 < 3(d_2 - 1)$, so both bounds are at most $d_2 - 1$. Theorem~\ref{thm:drm_lower}
already gives $r_{\mathcal{P}_1, \mathcal{P}_2}(k : (d_1, d_2)) \geq d_2 - 1$ for every $q$. Hence, no
improved bound over $\mathbb{F}_q$ with $q \geq 3$ follows from three vectors alone. In
Example~\ref{ex:q3}, the bound $d_2 - 1 = 4$ is attained. A bound of this nature for $q \geq 3$ requires more than three vectors. The Plotkin bound for
$D$-codes of Remark~\ref{rem:plotkin} and the linear programming bound of Section~\ref{sec:lp}
apply over any $\mathbb{F}_q$. These are the bounds we use in the non-binary case.
\end{remark}


\section{Linear Programming Bound}
\label{sec:lp}

In this section, we derive a linear programming (LP) lower bound on the optimal redundancy of generalized FCPCs over $\mathbb{F}_q$. Classical linear programming bounds for error-correcting codes constrain the size of a code of given length and minimum distance \cite[Th.~5.3.4]{vL1999}.
The bounds we present for GFCPCs run in the opposite direction. The message length is fixed and
the linear program constrains the block length, which gives a bound on the redundancy. Section~\ref{subsec:lp-prelim} makes this reformulation precise in the
error-correcting case, before any partition structure is imposed.

We first develop the bound for a single partition, which corresponds to a function-correcting code. The key ingredient is a MacWilliams-type positivity property for the distance distribution restricted to pairs lying in the same block of a partition. We then extend the bound to two partitions. By simultaneously tracking the distance distributions associated with the partitions $\mathcal{P}_1$ and $\mathcal{P}_2$, the two-partition linear programming bound can be strictly stronger than the one obtained from either partition alone, and we exhibit an example where this strict improvement occurs. Finally, we give the bound for an arbitrary number of partitions.

The bound applies to generalized FCPCs. The single-partition case covers any FCPC for a single partition, and hence any FCC \cite{LBZY2023}, and the two-partition case with $\mathcal{P}_1$ the finest partition covers any FCC with data protection \cite{RRHH20251}. For this reason we present the linear programming bound separately from the combinatorial bounds of Section~\ref{sec:bounds}. 

\subsection{Distance Distributions and MacWilliams-Type Inequalities}
\label{subsec:lp-dist}

Let $C \subseteq \mathbb{Z}_q^n$ be a code with $M = |C|$. We recall from Section~\ref{subsec:lp-prelim} the Krawtchouk polynomials $K_j(i)$ and the distance distribution $(B_0, B_1, \dots, B_n)$ of $C$. The \emph{Krawtchouk transform} of a sequence $(a_0, a_1, \dots, a_n)$ of real numbers is the sequence whose $j$-th entry is $\sum_{i=0}^{n} a_i\, K_j(i)$, for $j \in \{0, 1, \dots, n\}$. The linear programming bound rests on a positivity property of the Krawtchouk transform of the distance distribution, which we establish following the argument for error-correcting codes in \cite[Ch. 5]{vL1999}.
For the two lemmas of this subsection, we take the alphabet to be the ring $\mathbb{Z}_q = \mathbb{Z}/q\mathbb{Z}$, as in \cite[Ch. 5]{vL1999}. This is without loss of generality.
Let $\sigma : \mathbb{F}_q \to \mathbb{Z}_q$ be a bijection, extended coordinatewise to $\sigma : \mathbb{F}_q^n \to \mathbb{Z}_q^n$. For $x, y \in \mathbb{F}_q^n$ and $l \in [n]$ we have $x_l = y_l$ if and only if $\sigma(x_l) = \sigma(y_l)$, so $d(\sigma(x), \sigma(y)) = d(x, y)$. By a partition of $C$ into \emph{classes} we mean any partition $C_1, \dots, C_E$ of $C$ into
non-empty subsets.
Since $\sigma$ is a bijection on $\mathbb{F}_q^n$, a code $C \subseteq \mathbb{F}_q^n$ with classes $C_1, \dots, C_E$ has image $\sigma(C)$ with $|\sigma(C)| = |C|$ and $|\sigma(C_t)| = |C_t|$ for every $t$. 
The distance distributions used in this section are defined by counting ordered pairs of codewords at a given distance subject to conditions on the classes containing them, so they are unchanged by $\sigma$. The two lemmas are statements about these distributions, and therefore proving them over $\mathbb{Z}_q$ proves them over $\mathbb{F}_q$. Working over $\mathbb{Z}_q$ lets us write $\omega^{u \cdot x}$ for a $q^{\text{th}}$ root of unity $\omega$, since the exponent is then an element of $\mathbb{Z}_q$, which is used in the proof of Lemma \ref{lem:char-sum} and Lemma \ref{lem:macwilliams-block}.
Let $\omega$ be a primitive $q$th root of unity in $\mathbb{C}$. For $u \in \mathbb{Z}_q^n$, let $\omega^{u \cdot x}$ denote the value at $x \in \mathbb{Z}_q^n$ of the additive character indexed by $u$, where $u \cdot x$ is the inner product over $\mathbb{Z}_q$. We use that $\omega^{-a} = \overline{\omega^{a}}$ for $a \in \mathbb{Z}_q$. We first recall the identity relating these characters to the Krawtchouk polynomials
of~\eqref{eq:krawtchouk}.

\begin{lemma}[{\cite[Lemma 5.3.1]{vL1999}}]
\label{lem:char-sum}
For $v \in \mathbb{Z}_q^n$ with $\mathrm{wt}(v) = i$ and any $s \in \{0, 1, \dots, n\}$,
\[
\sum_{\substack{u \in \mathbb{Z}_q^n \\ \mathrm{wt}(u) = s}} \omega^{u \cdot v} = K_s(i).
\]
\end{lemma}

Let $C_1, \dots, C_E$ be a partition of $C$ into $E$ non-empty classes. Define the \emph{same-block distance distribution} $(B_0^{\emptyset}, \dots, B_n^{\emptyset})$ by
\begin{equation}
B_i^{\emptyset} = \frac{1}{M}\,\bigl|\{(x, y) \in C \times C : d(x, y) = i,\ x \text{ and } y \text{ lie in the same class}\}\bigr|.
\label{eq:Bempty}
\end{equation}

Every ordered pair of codewords at distance $i$ either lies entirely within one class or does not, so the pairs that are split across two distinct classes are counted by
\begin{equation}
B_i - B_i^{\emptyset} = \frac{1}{M}\,\bigl|\{(x, y) \in C \times C : d(x, y) = i,\ x \text{ and } y \text{ lie in different classes}\}\bigr|.
\label{eq:Bdiff}
\end{equation}
The following lemma will be the basis of our proposed bound.
\begin{lemma}
\label{lem:macwilliams-block}
Let $C \subseteq \mathbb{Z}_q^n$ be a code with $M = |C|$, let $C_1, \dots, C_E$ be a partition of $C$ into classes, and let $(B_i^{\emptyset})_{i=0}^{n}$ be given by \eqref{eq:Bempty}. Then for each $s \in \{0, 1, \dots, n\}$,
\[
\sum_{i=0}^{n} B_i^{\emptyset}\, K_s(i) \;\ge\; 0.
\]
\end{lemma}
\begin{proof}
By definition, $M B_i^{\emptyset} = \sum_{t=1}^E \bigl|\{(x, y) \in C_t \times C_t : d(x, y) = i\}\bigr|$. Multiplying by $K_s(i)$ and summing over $i$ gives
\begin{align*}
M \sum_{i=0}^{n} B_i^{\emptyset}\, K_s(i) &=  \sum_{i=0}^{n} K_s(i)\sum_{t=1}^E \bigl|\{(x, y) \in C_t \times C_t : d(x, y) = i\}\bigr|\\
&= \sum_{t=1}^E \sum_{i=0}^{n} K_s(i) \bigl|\{(x, y) \in C_t \times C_t : d(x, y) = i\}\bigr|\\
&=\sum_{t=1}^E \sum_{x, y \in C_t} K_s\bigl(d(x, y)\bigr).
\end{align*}
Since $d(x, y) = \mathrm{wt}(x - y)$, Lemma~\ref{lem:char-sum} gives $K_s\bigl(d(x, y)\bigr) = \displaystyle\sum_{\substack{u \in \mathbb{Z}_q^n \\ \mathrm{wt}(u) = s}} \omega^{u \cdot (x - y)}$. Substituting and interchanging the order of summation,
\[
M \sum_{i=0}^{n} B_i^{\emptyset}\, K_s(i) = \sum_{t=1}^E \sum_{\substack{u \in \mathbb{Z}_q^n \\ \mathrm{wt}(u) = s}} \sum_{x, y \in C_t} \omega^{u \cdot x}\,\overline{\omega^{u \cdot y}} = \sum_{t=1}^E \sum_{\substack{u \in \mathbb{Z}_q^n \\ \mathrm{wt}(u) = s}} \Bigl| \sum_{x \in C_t} \omega^{u \cdot x} \Bigr|^{2},
\]
which is nonnegative.
\end{proof}

Taking $E = 1$, so that the single class is $C$ itself and $B_i^{\emptyset} = B_i$, recovers the inequalities of Lemma~\ref{thm:macwilliams}. By the identification made above, Lemma~\ref{lem:char-sum} and Lemma~\ref{lem:macwilliams-block} apply to codes over $\mathbb{F}_q$, and we use them in that form for the rest of this section.
Now let $\mathcal{P}$ be a partition of $\mathbb{F}_q^k$ with block count $E = |\mathcal{P}|$, let $\mathcal{C} : \mathbb{F}_q^k \to \mathbb{F}_q^{n}$ be a systematic $(\mathcal{P} : d)$-encoding, and let $C = \mathrm{Im}(\mathcal{C})$ be its image. The partition $\mathcal{P}$ induces a partition of $C$ into classes $C_1, \dots, C_{E}$, where $C_t$ is the set of codewords whose corresponding messages lie in the $t$-th block of $\mathcal{P}$. Since the encoding is systematic, distinct messages produce distinct codewords, so $|C_t| = m_t$, where $m_t$ is the size of the $t$-th block of $\mathcal{P}$, and $M = \sum_{t} m_t = q^k$. A codeword determines its message, so for codewords $x, y \in C$ we write $\mathcal{P}(x) = \mathcal{P}(y)$ to mean that their messages lie in the same block of $\mathcal{P}$, and $\mathcal{P}(x) \ne \mathcal{P}(y)$ otherwise.
We now establish two counting identities that relate the full and same-block distance distributions. Summing the full distance distribution over all $i$ counts each ordered pair of codewords exactly once,
\begin{equation}
\sum_{i=0}^{n} B_i = \frac{1}{M} \sum_{i=0}^{n} \bigl|\{(x, y) \in C \times C : d(x, y) = i\}\bigr| = \frac{1}{M}\,|C \times C| = \frac{M^2}{M} = M.
\label{eq:sum-B}
\end{equation}
Summing the same-block distance distribution over all $i$ counts each ordered pair of codewords whose members lie in a common class,
\begin{equation}
\sum_{i=0}^{n} B_i^{\emptyset} = \frac{1}{M} \bigl|\{(x, y) \in C \times C : x, y \text{ lie in the same class}\}\bigr| = \frac{1}{M} \sum_{t=1}^E |C_t|^2 = \frac{1}{M} \sum_{t=1}^E m_t^2,
\label{eq:sum-Bempty}
\end{equation}
where the third equality holds because the classes $C_1, \dots, C_{E}$ are disjoint and the ordered pairs with both members in $C_t$ number $|C_t|^2 = m_t^2$.
Multiplying \eqref{eq:sum-Bempty} by $M^2$ gives $M^2 \sum_{i} B_i^{\emptyset} = M \sum_{t} m_t^2$. Dividing by $\sum_{t} m_t^2$ and using $\sum_{i} B_i = M$ from \eqref{eq:sum-B} yields the \emph{partition condition}
\begin{equation}
\sum_{i=0}^{n} B_i = \widetilde{E}(\mathcal{P}) \sum_{i=0}^{n} B_i^{\emptyset}, \qquad \text{where} \ \ \widetilde{E}(\mathcal{P}) = \frac{M^2}{\sum_t m_t^2} = \frac{\bigl(\sum_t m_t\bigr)^2}{\sum_t m_t^2}.
\label{eq:partition-cond}
\end{equation}
We call $\widetilde{E}(\mathcal{P})$ the \emph{effective number of blocks} of $\mathcal{P}$. The tilde distinguishes it from the block count $E = |\mathcal{P}|$, which it need not equal. It depends only on the block-size profile of $\mathcal{P}$ and is invariant under a common scaling of the block sizes.

\begin{remark}
\label{rem:eff-vs-blocks}
By the Cauchy--Schwarz inequality, $\bigl(\sum_t m_t\bigr)^2 \le E \sum_t m_t^2$, so
\[
\widetilde{E}(\mathcal{P}) \le E,
\]
with equality if and only if all blocks of $\mathcal{P}$ have the same size. We call such a partition \emph{balanced}. The partition condition \eqref{eq:partition-cond} is therefore an equality with the effective number of blocks $\widetilde{E}(\mathcal{P})$ and not with the block count $E$. The two agree exactly in the balanced case.
\end{remark}

\subsection{Linear Programming Bound for a Single Partition}
\label{subsec:lp-single}

Let $\mathcal{P}$ be a partition of $\mathbb{F}_q^k$ with effective number of blocks $\widetilde{E}(\mathcal{P})$, and fix a length $n$ and a distance requirement $d$. Consider the following linear program in the variables $(B_i)_{i=0}^n$ and $(B_i^{\emptyset})_{i=0}^n$:

\begin{linearprogram}\label{lp:single}
\begin{equation}
M_{\mathrm{LP}}(n, d, \widetilde{E}(\mathcal{P})) = \max \ \sum_{i=0}^{n} B_i
\label{eq:lp-single-obj}
\end{equation}
subject to
\begin{align}
&B_0 = B_0^{\emptyset} = 1, \label{eq:lp1}\\
&B_i^{\emptyset} \ge 0 \ \text{ and } \ B_i - B_i^{\emptyset} \ge 0, && i = 0, \dots, n, \label{eq:lp2}\\
&B_i - B_i^{\emptyset} = 0, && 1 \le i \le d - 1, \label{eq:lp3}\\
&\sum_{i=0}^{n} B_i\, K_j(i) \ge 0, && j = 0, \dots, n, \label{eq:lp4}\\
&\sum_{i=0}^{n} B_i^{\emptyset}\, K_j(i) \ge 0, && j = 0, \dots, n, \label{eq:lp5}\\
&\sum_{i=0}^{n} B_i = \widetilde{E}(\mathcal{P}) \sum_{i=0}^{n} B_i^{\emptyset}. \label{eq:lp6}
\end{align}
\end{linearprogram}
Constraint \eqref{eq:lp3} is the \emph{distance constraint}, since codewords whose messages lie in different blocks of $\mathcal{P}$ must be at Hamming distance at least $d$. Constraints \eqref{eq:lp4} and \eqref{eq:lp5} are the Delsarte inequalities of Lemma~\ref{thm:macwilliams} for the full code and of Lemma~\ref{lem:macwilliams-block} for the same-block distribution, respectively. Constraint \eqref{eq:lp6} is the partition condition \eqref{eq:partition-cond}. If the program has no feasible solution, which occurs for small $n$, we set
$M_{\mathrm{LP}}(n, d, \widetilde{E}(\mathcal{P})) = -\infty$. The same convention applies to the
programs of Sections~\ref{subsec:lp-two} and \ref{subsec:lp-multi}.

\begin{remark}
\label{rem:ecc-reduction}
Let $\mathcal{P}$ be the finest partition of $\mathbb{F}_q^k$, whose blocks are singleton sets, so that $\widetilde{E}(\mathcal{P}) = q^k$. Distinct messages lie in distinct blocks, so the image of a $(\mathcal{P} : d)$-encoding has $B_i^{\emptyset} = 0$ for $i \ge 1$ and same-block total $\sum_{i} B_i^{\emptyset} = 1$. Substituting these values, the distance constraint \eqref{eq:lp3} gives $B_i = 0$ for $1 \le i \le d-1$, the same-block Delsarte inequality \eqref{eq:lp5} holds trivially, and the program LP\ref{lp:single} becomes the linear program of \cite[Th.~5.3.4]{vL1999} for error-correcting codes, except for the
partition condition \eqref{eq:lp6}, which must be dropped.

Dropping \eqref{eq:lp6} is what the finest partition requires. Its purpose is to tie the same-block total $\sum_i B_i^{\emptyset}$ to the size of the blocks of $\mathcal{P}$, but once the blocks are singletons this total is already fixed to $1$ by $B_i^{\emptyset} = 0$ for $i \ge 1$. The condition \eqref{eq:lp6} would then read $\sum_i B_i = q^k$, fixing the objective, whereas the linear program of \cite[Th.~5.3.4]{vL1999} maximizes $\sum_i B_i$. Removing \eqref{eq:lp6} restores the maximization.

\end{remark}

The optimal redundancy of an FCPC is lower bounded as follows.

\begin{theorem}
\label{thm:lp-single}
Let $\mathcal{P}$ be a partition of $\mathbb{F}_q^k$ with effective number of blocks $\widetilde{E}(\mathcal{P})$, and let $d$ be a distance requirement. Then
\[
r_{\mathcal{P}}(k : d) \ge \min\bigl\{\, n' : M_{\mathrm{LP}}(n', d, \widetilde{E}(\mathcal{P})) \ge q^k \,\bigr\} - k.
\]
\end{theorem}

\begin{proof}
Let $\mathcal{C} : \mathbb{F}_q^k \to \mathbb{F}_q^{k + r}$ be an optimal $(\mathcal{P} : d)$-encoding, so $r = r_{\mathcal{P}}(k : d)$ and $n = k + r$. Its image $C = \mathrm{Im}(\mathcal{C})$ is a code of length $n$ with $|C| = q^k$. 
The distance distributions $(B_i)_{i=0}^{n}$ and $(B_i^{\emptyset})_{i=0}^{n}$ of $C$ satisfy \eqref{eq:lp1}--\eqref{eq:lp6}. Constraint \eqref{eq:lp3} holds because $\mathcal{C}$ is a $(\mathcal{P} : d)$-encoding, constraints \eqref{eq:lp4}--\eqref{eq:lp5} hold by Lemma~\ref{thm:macwilliams} and Lemma~\ref{lem:macwilliams-block}, and \eqref{eq:lp6} holds by \eqref{eq:partition-cond}. Hence, the distributions $(B_i)_{i=0}^{n}$ and $(B_i^{\emptyset})_{i=0}^{n}$ are feasible for the program, and
\[
q^k = \sum_{i=0}^{n} B_i \le M_{\mathrm{LP}}(n, d, \widetilde{E}(\mathcal{P})).
\]
Therefore, any length $n$ that admits such an encoding satisfies $M_{\mathrm{LP}}(n, d, \widetilde{E}(\mathcal{P})) \ge q^k$, so $$n \ge \min\{n' : M_{\mathrm{LP}}(n', d, \widetilde{E}(\mathcal{P})) \ge q^k\},$$ and $r = n - k$ gives the stated bound.
\end{proof}

We illustrate Theorem~\ref{thm:lp-single} with two examples. Example~\ref{ex6} considers a linear function, which corresponds to a balanced partition, and Example~\ref{ex:1P_UB} considers the Hamming weight function, which corresponds to an unbalanced partition. In both cases, we compare the resulting bound with the Hamming bound and the Plotkin bound for FCCs.

\begin{example}
\label{ex6}
Let $f : \mathbb{F}_2^4 \to \mathbb{F}_2^2$ be the linear function
\[
f(x) = \begin{pmatrix} 1 & 1 & 1 & 0 \\ 0 & 1 & 1 & 0 \end{pmatrix} x .
\]
Its domain partition $\mathcal{P}_f$ has four blocks,
\begin{align*}
f^{-1}(00) &= \{0000, 0001, 0110, 0111\}, & f^{-1}(01) &= \{1010, 1011, 1100, 1101\},\\
f^{-1}(10) &= \{1000, 1001, 1110, 1111\}, & f^{-1}(11) &= \{0010, 0011, 0100, 0101\},
\end{align*}
each of size $4$. The partition is balanced, so the effective number of blocks agrees with the
block count, $\widetilde{E}(\mathcal{P}_f) = E = 4$.

For each distance requirement $d$, we solve LP\ref{lp:single} at successive lengths until we find the
smallest length $n_d$ at which the optimum is greater than or equal to $2^4$. Theorem~\ref{thm:lp-single} then gives
$r_{\mathcal{P}_f}(4 : d) \geq n_d - 4$. Table~\ref{tab:ex6-threshold} lists the optimal values at
$n_d$ and at the preceding length. 

\begin{table}[!t]
\centering
\caption{Optimal values of LP\ref{lp:single} in Example~\ref{ex6}, where $n_d$ is the smallest length
with $M_{\mathrm{LP}}\bigl(n, d, \widetilde{E}(\mathcal{P}_f)\bigr) \geq 2^4$.}
\label{tab:ex6-threshold}
\renewcommand{\arraystretch}{1.15}
\setlength{\tabcolsep}{9pt}
\begin{tabular}{c|cc|cc}
\hline
& \multicolumn{2}{c|}{$n = n_d - 1$} & \multicolumn{2}{c}{$n = n_d$} \\
$d$ & $n$ & $M_{\mathrm{LP}}(n, d, 4)$ & $n$ & $M_{\mathrm{LP}}(n, d, 4)$ \\
\hline
$3$  & $5$  & $4$ & $6$  & $16$      \\
$5$  & $8$  & $4$ & $9$  & $20$      \\
$7$  & $11$ & $4$ & $12$ & $1408/27$ \\
$9$  & $14$ & $4$ & $15$ & $64$      \\
$11$ & $17$ & $4$ & $18$ & $1440/17$ \\
$13$ & $20$ & $4$ & $21$ & $3740/43$ \\
\hline
\end{tabular}
\end{table}

Table~\ref{tab:ex6-bounds} compares the resulting bounds with the Hamming bound for FCCs
of \cite{RRHH20251} and the Plotkin bound of \cite{LBZY2023}.

\begin{table}[!t]
\centering
\caption{Lower bounds on the optimal redundancy in Example~\ref{ex6}.}
\label{tab:ex6-bounds}
\renewcommand{\arraystretch}{1.15}
\setlength{\tabcolsep}{9pt}
\begin{tabular}{c|ccc}
\hline
& \multicolumn{3}{c}{Lower bound on $r_{\mathcal{P}_f}(4 : d)$} \\
$d$ & Hamming \cite{RRHH20251} & Plotkin \cite{LBZY2023} & LP \\
\hline
$3$  & $r \geq 3$  & $r \geq 2$  & $r \geq 2$  \\
$5$  & $r \geq 5$  & $r \geq 5$  & $r \geq 5$  \\
$7$  & $r \geq 7$  & $r \geq 8$  & $r \geq 8$  \\
$9$  & $r \geq 10$ & $r \geq 11$ & $r \geq 11$ \\
$11$ & $r \geq 12$ & $r \geq 14$ & $r \geq 14$ \\
$13$ & $r \geq 14$ & $r \geq 17$ & $r \geq 17$ \\
\hline
\end{tabular}
\end{table}

At $d = 3$, the Hamming bound is the strongest of the three. At $d = 5$ all three agree. For every
larger value of $d$ in the table the linear programming bound strictly improves on the Hamming bound,
and the gap grows with $d$. The linear programming
bound and the Plotkin bound coincide at every value of $d$ listed here. 
\end{example}

Example~\ref{ex:1P_UB} considers an unbalanced partition, for which the linear programming bound
improves on the Hamming bound at every distance requirement listed and is strictly stronger than
the Plotkin bound at some values.

\begin{example}\label{ex:1P_UB}
Consider the weight function $g : \mathbb{F}_2^5 \to \{0, 1, 2, 3, 4, 5\}$ defined by
\[
g(u) = \mathrm{wt}(u),
\]
for every $u \in \mathbb{F}_2^5$. Its domain partition $\mathcal{P}_g$ has six blocks,
\begin{align*}
g^{-1}(0) &= \{00000\},\\
g^{-1}(1) &= \{10000, 01000, 00100, 00010, 00001\},\\
g^{-1}(2) &= \{11000, 10100, 10010, 10001, 01100, 01010, 01001, 00110, 00101, 00011\},\\
g^{-1}(3) &= \{00111, 01011, 01101, 01110, 10011, 10101, 10110, 11001, 11010, 11100\},\\
g^{-1}(4) &= \{01111, 10111, 11011, 11101, 11110\},\\
g^{-1}(5) &= \{11111\},
\end{align*}
of sizes $1, 5, 10, 10, 5, 1$. The partition is not balanced, and the effective number of blocks
$\widetilde{E}(\mathcal{P}_g) = 256/63$ is less than the block count $E = 6$.

Solving LP\ref{lp:single} gives the values in
Table~\ref{tab:weight-threshold}, and the resulting bounds appear in Table~\ref{tab:weight-bounds}.
Here the program has no feasible solution at $n = n_d - 1$ for $d \geq 7$. At $n = n_d - 1$, the program has no feasible solution for $d \geq 7$, so the entry is $-\infty$. For
$d = 3$ and $d = 5$, a feasible solution exists but is less than $2^5$.

\begin{table}[!t]
\centering
\caption{Optimal values of LP\ref{lp:single} in Example~\ref{ex:1P_UB}, where $n_d$ is the smallest
length with $M_{\mathrm{LP}}\bigl(n, d, \widetilde{E}(\mathcal{P}_g)\bigr) \geq 2^5$.}
\label{tab:weight-threshold}
\renewcommand{\arraystretch}{1.15}
\setlength{\tabcolsep}{9pt}
\begin{tabular}{c|cc|cc}
\hline
& \multicolumn{2}{c|}{$n = n_d - 1$} & \multicolumn{2}{c}{$n = n_d$} \\
$d$ & $n$ & $M_{\mathrm{LP}}(n, d, 256/63)$ & $n$ & $M_{\mathrm{LP}}(n, d, 256/63)$ \\
\hline
$3$  & $6$  & $512/33$  & $7$  & $188416/3285$ \\
$5$  & $9$  & $256/13$  & $10$ & $768/11$      \\
$7$  & $11$ & $-\infty$ & $12$ & $352/7$       \\
$9$  & $14$ & $-\infty$ & $15$ & $12032/231$   \\
$11$ & $17$ & $-\infty$ & $18$ & $256/5$       \\
$13$ & $20$ & $-\infty$ & $21$ & $128/3$       \\
\hline
\end{tabular}
\end{table}

Table~\ref{tab:weight-bounds} compares the resulting bounds with the Hamming bound for FCCs
of \cite{RRHH20251} and the Plotkin bound of \cite{LBZY2023}.

\begin{table}[!t]
\centering
\caption{Lower bounds on the optimal redundancy in Example~\ref{ex:1P_UB}.}
\label{tab:weight-bounds}
\renewcommand{\arraystretch}{1.15}
\setlength{\tabcolsep}{9pt}
\begin{tabular}{c|ccc}
\hline
& \multicolumn{3}{c}{Lower bound on $r_{\mathcal{P}_g}(5 : d)$} \\
$d$ & Hamming \cite{RRHH20251} & Plotkin \cite{LBZY2023} & LP \\
\hline
$3$  & $r \geq 1$  & $r \geq 1$  & $r \geq 2$  \\
$5$  & $r \geq 3$  & $r \geq 4$  & $r \geq 5$  \\
$7$  & $r \geq 6$  & $r \geq 7$  & $r \geq 7$  \\
$9$  & $r \geq 8$  & $r \geq 10$ & $r \geq 10$ \\
$11$ & $r \geq 10$ & $r \geq 13$ & $r \geq 13$ \\
$13$ & $r \geq 12$ & $r \geq 16$ & $r \geq 16$ \\
\hline
\end{tabular}
\end{table}

The linear programming bound is strictly stronger than the Hamming bound at every distance
requirement in the table. It is also strictly stronger than the Plotkin bound at $d = 3$ and
$d = 5$, and the two agree for the larger values of $d$. 
\end{example}

\subsection{Linear Programming Bound for Two Partitions}
\label{subsec:lp-two}

We now consider two partitions $\mathcal{P}_1$ and $\mathcal{P}_2$ of $\mathbb{F}_q^k$ with distance requirements $d_1 \le d_2$. The refinement of the bound comes from splitting the distance distribution according to which of the two partitions separates a given pair. For the image $C$ of a systematic $(\mathcal{P}_h : d_h;\, h \in [2])$-FCPC and each $i \in \{0, \dots, n\}$, define, using the convention of Section~\ref{subsec:lp-dist} for the action of a partition on codewords,
\begin{align}
B_i^{\emptyset} &= \tfrac{1}{M}\,\bigl|\{(x, y) \in C^2 : d(x, y) = i,\ \mathcal{P}_1(x) = \mathcal{P}_1(y),\ \mathcal{P}_2(x) = \mathcal{P}_2(y)\}\bigr|, \label{eq:two-empty}\\
B_i^{1} &= \tfrac{1}{M}\,\bigl|\{(x, y) \in C^2 : d(x, y) = i,\ \mathcal{P}_1(x) \ne \mathcal{P}_1(y),\ \mathcal{P}_2(x) = \mathcal{P}_2(y)\}\bigr|, \label{eq:two-1}\\
B_i^{2} &= \tfrac{1}{M}\,\bigl|\{(x, y) \in C^2 : d(x, y) = i,\ \mathcal{P}_1(x) = \mathcal{P}_1(y),\ \mathcal{P}_2(x) \ne \mathcal{P}_2(y)\}\bigr|, \label{eq:two-2}\\
B_i^{1,2} &= \tfrac{1}{M}\,\bigl|\{(x, y) \in C^2 : d(x, y) = i,\ \mathcal{P}_1(x) \ne \mathcal{P}_1(y),\ \mathcal{P}_2(x) \ne \mathcal{P}_2(y)\}\bigr|. \label{eq:two-12}
\end{align}
The superscript represents the set of partitions in which the pair differs. The four distributions are disjoint and sum to the full distance distribution, $B_i = B_i^{\emptyset} + B_i^{1} + B_i^{2} + B_i^{1,2}$.

Four combinations of these distributions are same-block distance distributions, one for each of the partitions $\mathcal{P}_1$, $\mathcal{P}_2$, their join $\mathcal{P}_1 \vee \mathcal{P}_2$, and the trivial one-block partition:
\begin{align*}
\text{whole code:} && B_i^{\emptyset} + B_i^{1} + B_i^{2} + B_i^{1,2} &= B_i,\\[2pt]
\text{same block of } \mathcal{P}_1: && B_i^{\emptyset} + B_i^{2}, &\\[2pt]
\text{same block of } \mathcal{P}_2: && B_i^{\emptyset} + B_i^{1}, &\\[2pt]
\text{same block of } \mathcal{P}_1 \vee \mathcal{P}_2: && B_i^{\emptyset}. &
\end{align*}
Each of these is the same-block distance distribution of $C$ with respect to one of the partitions $\mathcal{P}_1$, $\mathcal{P}_2$, $\mathcal{P}_1 \vee \mathcal{P}_2$, or the trivial one-block partition, so Lemma~\ref{lem:macwilliams-block} applies to each and shows that its Krawtchouk transform is nonnegative. We emphasize that this positivity holds only for the four same-block combinations. A single distribution $B_i^{S}$ with $S \ne \emptyset$, or a combination corresponding to differing in a partition, such as $B_i^{1} + B_i^{1,2}$ for the pairs differing in $\mathcal{P}_1$, is not a same-block distribution, and its Krawtchouk transform need not be sign-definite. Delsarte inequalities may therefore be imposed only on the four same-block combinations above.

A pair differing in $\mathcal{P}_1$ must be separated by distance $d_1$, and a pair differing in $\mathcal{P}_2$ by distance $d_2$. Since $d_1 \le d_2$, the pairs counted by $B_i^{1}$ (differing in $\mathcal{P}_1$ only) require distance $d_1$, while the pairs counted by $B_i^{2}$ and $B_i^{1,2}$ (differing in $\mathcal{P}_2$) require distance $d_2$.

Fix a length $n$ and distances $d_1 \le d_2$. Consider the linear program in the variables $(B_i^{\emptyset}, B_i^{1}, B_i^{2}, B_i^{1,2})_{i=0}^{n}$:
\begin{linearprogram}\label{lp:two}
\begin{equation}
M_{\mathrm{LP}}\bigl(n, d_1, d_2, \widetilde{E}(\mathcal{P}_1), \widetilde{E}(\mathcal{P}_2), \widetilde{E}(\mathcal{P}_1 \vee \mathcal{P}_2)\bigr) = \max \ \sum_{i=0}^{n} \bigl(B_i^{\emptyset} + B_i^{1} + B_i^{2} + B_i^{1,2}\bigr)
\label{eq:lp-two-obj}
\end{equation}
subject to
\begin{align}
&B_0^{\emptyset} = 1, \quad B_0^{1} = B_0^{2} = B_0^{1,2} = 0, \label{eq:tlp1}\\
&B_i^{\emptyset},\, B_i^{1},\, B_i^{2},\, B_i^{1,2} \ge 0, && i = 0, \dots, n, \label{eq:tlp2}\\
&B_i^{1} = 0, && 1 \le i \le d_1 - 1, \label{eq:tlp3a}\\
&B_i^{2} = B_i^{1,2} = 0, && 1 \le i \le d_2 - 1, \label{eq:tlp3b}\\
&\sum_{i=0}^{n} B_i\, K_j(i) \ge 0, && j = 0, \dots, n, \label{eq:tlp4a}\\
&\sum_{i=0}^{n} \bigl(B_i^{\emptyset} + B_i^{2}\bigr) K_j(i) \ge 0, && j = 0, \dots, n, \label{eq:tlp4b}\\
&\sum_{i=0}^{n} \bigl(B_i^{\emptyset} + B_i^{1}\bigr) K_j(i) \ge 0, && j = 0, \dots, n, \label{eq:tlp4c}\\
&\sum_{i=0}^{n} B_i^{\emptyset}\, K_j(i) \ge 0, && j = 0, \dots, n, \label{eq:tlp4d}\\
&\sum_{i=0}^{n} B_i = \widetilde{E}(\mathcal{P}_1) \sum_{i=0}^{n} \bigl(B_i^{\emptyset} + B_i^{2}\bigr) \notag\\
&\hphantom{\sum_{i=0}^{n} B_i} = \widetilde{E}(\mathcal{P}_2) \sum_{i=0}^{n} \bigl(B_i^{\emptyset} + B_i^{1}\bigr) \notag\\
&\hphantom{\sum_{i=0}^{n} B_i} = \widetilde{E}(\mathcal{P}_1 \vee \mathcal{P}_2) \sum_{i=0}^{n} B_i^{\emptyset}. \label{eq:tlp5}
\end{align}
\end{linearprogram}
Constraints \eqref{eq:tlp3a} and \eqref{eq:tlp3b} are the \emph{distance constraints}. Constraints \eqref{eq:tlp4a}--\eqref{eq:tlp4d} are the Delsarte inequalities for the whole code, the same-$\mathcal{P}_1$, the same-$\mathcal{P}_2$, and the same-join distributions. The chain of equalities \eqref{eq:tlp5} is the partition condition \eqref{eq:partition-cond} applied to $\mathcal{P}_1$, $\mathcal{P}_2$, and $\mathcal{P}_1 \vee \mathcal{P}_2$, and the common left-hand side is $\sum_i B_i = M$.

\begin{theorem}
\label{thm:lp-two}
Let $\mathcal{P}_1$ and $\mathcal{P}_2$ be partitions of $\mathbb{F}_q^k$ with distance requirements $d_1 \le d_2$. Then
\[
r_{\mathcal{P}_1, \mathcal{P}_2}(k : (d_1, d_2)) \ge \min\bigl\{\, n' : M_{\mathrm{LP}}\bigl(n', d_1, d_2, \widetilde{E}(\mathcal{P}_1), \widetilde{E}(\mathcal{P}_2), \widetilde{E}(\mathcal{P}_1 \vee \mathcal{P}_2)\bigr) \ge q^k \,\bigr\} - k.
\]
\end{theorem}

\begin{proof}
Let $\mathcal{C} : \mathbb{F}_q^k \to \mathbb{F}_q^{k + r}$ be an optimal $(\mathcal{P}_h : d_h;\, h \in [2])$-FCPC, so $r = r_{\mathcal{P}_1, \mathcal{P}_2}(k : (d_1, d_2))$ and $n = k + r$, and let $C = \mathrm{Im}(\mathcal{C})$. The distance distributions \eqref{eq:two-empty}--\eqref{eq:two-12} of $C$ satisfy \eqref{eq:tlp1}--\eqref{eq:tlp5}. The distance constraints \eqref{eq:tlp3a}--\eqref{eq:tlp3b} hold because $\mathcal{C}$ separates $\mathcal{P}_1$-differing pairs by $d_1$ and $\mathcal{P}_2$-differing pairs by $d_2$. The Delsarte inequalities \eqref{eq:tlp4a}--\eqref{eq:tlp4d} hold by Lemma~\ref{lem:macwilliams-block} applied to the partitions of $C$ induced by $\mathcal{P}_1$, $\mathcal{P}_2$, $\mathcal{P}_1 \vee \mathcal{P}_2$, and the trivial one-block partition. The chain \eqref{eq:tlp5} holds by the partition condition \eqref{eq:partition-cond}. Feasibility gives $q^k = \sum_i B_i \le M_{\mathrm{LP}}\bigl(n, d_1, d_2, \widetilde{E}(\mathcal{P}_1), \widetilde{E}(\mathcal{P}_2), \widetilde{E}(\mathcal{P}_1 \vee \mathcal{P}_2)\bigr)$. The bound follows as in the proof of Theorem~\ref{thm:lp-single}.
\end{proof}

\begin{remark}
\label{rem:reduction}
If $\mathcal{P}_1$ is a refinement of $\mathcal{P}_2$, then every pair in the same block of $\mathcal{P}_1$ is also in the same block of $\mathcal{P}_2$, so $B_i^{2} = 0$ for all $i$. Consequently $\mathcal{P}_1 \vee \mathcal{P}_2 = \mathcal{P}_1$ and $\widetilde{E}(\mathcal{P}_1 \vee \mathcal{P}_2) = \widetilde{E}(\mathcal{P}_1)$, the same-$\mathcal{P}_1$ distribution $B_i^{\emptyset} + B_i^{2}$ equals the same-join distribution $B_i^{\emptyset}$, and \eqref{eq:tlp4b} duplicates \eqref{eq:tlp4d}. The first and third equalities of \eqref{eq:tlp5} then coincide as well. The program involves only the three distributions $B_i$, $B_i^{\emptyset} + B_i^{1}$, and $B_i^{\emptyset}$, retaining the whole-code, same-$\mathcal{P}_2$, and same-join Delsarte inequalities \eqref{eq:tlp4a}, \eqref{eq:tlp4c}, \eqref{eq:tlp4d} together with two of the three equalities in \eqref{eq:tlp5}. When $\mathcal{P}_1$ is moreover the finest partition, whose blocks are singletons, this is the data-protection case, where $\mathcal{P}_1$ enforces the distance requirement $d_1$ on all distinct message pairs~\cite{RRHH20251}. The caution of Remark~\ref{rem:ecc-reduction} applies here as well, since $\widetilde{E}(\mathcal{P}_1) = q^k$ and the constraints do not force $B_i^{\emptyset} = 0$ for $i \ge 1$.

\end{remark}

\begin{example}\label{ex9}
Let $\mathcal{P}_2 = \mathcal{P}_f$ be the domain partition of the linear function $f$ of
Example~\ref{ex6}, and let $\mathcal{P}_1$ be the finest partition of $\mathbb{F}_2^4$. By
Remark~\ref{rem:generalization}, a $(\mathcal{P}_h : d_h;\, h \in [2])$-FCPC is then an
$(f : d_1, d_2)$-FCC, so the bound of Theorem~\ref{thm:lp-two} applies to an FCC with data protection,
where $d_1$ is the distance for the data and $d_2$ the distance for the function value. The blocks of
$\mathcal{P}_1$ are singletons, so $\widetilde{E}(\mathcal{P}_1) = 16$, and
$\widetilde{E}(\mathcal{P}_2) = 4$ as in Example~\ref{ex6}. Since $\mathcal{P}_1$ refines
$\mathcal{P}_2$, the join $\mathcal{P}_1 \vee \mathcal{P}_2$ is $\mathcal{P}_1$ and
$\widetilde{E}(\mathcal{P}_1 \vee \mathcal{P}_2) = 16$, and LP\ref{lp:two} takes the reduced form
described in Remark~\ref{rem:reduction}.

Solving LP\ref{lp:two} gives the bounds in Table~\ref{tab:dp-bounds}.

\begin{table}[!t]
\centering
\caption{Lower bounds on the optimal redundancy in Example~\ref{ex9}.}
\label{tab:dp-bounds}
\renewcommand{\arraystretch}{1.15}
\setlength{\tabcolsep}{9pt}
\begin{tabular}{c|ccc}
\hline
& \multicolumn{3}{c}{Lower bound on $r_{\mathcal{P}_1, \mathcal{P}_2}(4 : (d_1, d_2))$} \\
$(d_1, d_2)$ & Hamming \cite{RRHH20251} & Plotkin \cite{LBZY2023} & LP \\
\hline
$(3, 5)$  & $r \geq 6$  & $r \geq 5$  & $r \geq 6$  \\
$(3, 7)$  & $r \geq 8$  & $r \geq 8$  & $r \geq 9$  \\
$(3, 9)$  & $r \geq 10$ & $r \geq 11$ & $r \geq 12$ \\
$(5, 11)$ & $r \geq 13$ & $r \geq 15$ & $r \geq 16$ \\
$(7, 11)$ & $r \geq 14$ & $r \geq 16$ & $r \geq 16$ \\
$(9, 11)$ & $r \geq 14$ & $r \geq 16$ & $r \geq 17$ \\
$(9, 13)$ & $r \geq 16$ & $r \geq 19$ & $r \geq 20$ \\
\hline
\end{tabular}
\end{table}

The linear programming bound is at least as strong as the Plotkin bound at every pair in the table,
and strictly stronger at all of them except $(7, 11)$. Comparing with Example~\ref{ex6} shows the
cost of adding data protection. For $d_2 = 5$, $7$, and $9$ the single-partition bound of
Example~\ref{ex6} gives $r \geq 5$, $8$, and $11$, while protecting the data at $d_1 = 3$ raises
these to $r \geq 6$, $9$, and $12$. 
\end{example}

In Example~\ref{ex9} the partition $\mathcal{P}_1$ refines $\mathcal{P}_2$. We now take two
partitions for which neither refines the other.

\begin{example}\label{ex:lp-two-order}
Consider the functions $g_1 : \mathbb{F}_3^4 \to \mathbb{F}_3$ and
$g_2 : \mathbb{F}_3^4 \to \mathbb{F}_3^2$ defined by
\[
g_1(u) = \sum_{i=1}^{4} u_i \pmod 3, \qquad
g_2(u) = (u_1 u_2,\ u_3 + u_4),
\]
for every $u = (u_1, u_2, u_3, u_4) \in \mathbb{F}_3^4$. The domain partition
$\mathcal{P}_{g_1}$ has three blocks,
\[\lvert g_1^{-1}(0)  \rvert= \lvert g_1^{-1}(1)  \rvert= \lvert g_1^{-1}(2)  \rvert =27, \]
and the domain partition $\mathcal{P}_{g_2}$ has nine blocks,
\begin{align*}
   \lvert g_2^{-1}(0,0)  \rvert&= \lvert g_2^{-1}(0,1)  \rvert= \lvert g_2^{-1}(0,2)  \rvert =15, \\
\lvert g_2^{-1}(1,0)  \rvert&= \lvert g_2^{-1}(1,1)  \rvert= \lvert g_2^{-1}(1,2)  \rvert= \lvert g_2^{-1}(2,0)  \rvert= \lvert g_2^{-1}(2,1)  \rvert= \lvert g_2^{-1}(2,2)  \rvert =6.
\end{align*}
Hence $\widetilde{E}(\mathcal{P}_{g_1}) = 3$ and
$\widetilde{E}(\mathcal{P}_{g_2}) = 81/11$. The join has block sizes
\[
(\underbrace{3,3,\ldots,3}_{9\text{ times}},\underbrace{6,6,\ldots,6}_{9\text{ times}})
\]
and
$\widetilde{E}(\mathcal{P}_{g_1} \vee \mathcal{P}_{g_2}) = 81/5$. Neither partition refines the
other, so LP\ref{lp:two} retains all four Delsarte inequalities and all three equalities
of~\eqref{eq:tlp5}.
Since we assume $d_1 \leq d_2$, the partition carrying the smaller requirement is the one labelled
$\mathcal{P}_1$. The two columns of Table~\ref{tab:two-order} are the two ways of doing this, so they are
two different sets of distance constraints and hence two different examples.

\begin{table}[!t]
\centering
\caption{Lower bounds on the optimal redundancy in Example~\ref{ex:lp-two-order} for the two
assignments of the distance requirements to $\mathcal{P}_{g_1}$ and $\mathcal{P}_{g_2}$.}
\label{tab:two-order}
\renewcommand{\arraystretch}{1.15}
\setlength{\tabcolsep}{9pt}
\begin{tabular}{c|cc|cc}
\hline
& \multicolumn{2}{c|}{$\mathcal{P}_1 = \mathcal{P}_{g_1}$, $\mathcal{P}_2 = \mathcal{P}_{g_2}$}
& \multicolumn{2}{c}{$\mathcal{P}_1 = \mathcal{P}_{g_2}$, $\mathcal{P}_2 = \mathcal{P}_{g_1}$} \\
$(d_1, d_2)$ & Plotkin \cite{LBZY2023} & LP & Plotkin \cite{LBZY2023} & LP \\
\hline
$(2, 3)$   & $r \geq 1$  & $r \geq 2$  & $r \geq 1$  & $r \geq 1$  \\
$(3, 3)$   & $r \geq 1$  & $r \geq 2$  & $r \geq 1$  & $r \geq 2$  \\
$(3, 4)$   & $r \geq 2$  & $r \geq 3$  & $r \geq 2$  & $r \geq 3$  \\
$(4, 4)$   & $r \geq 2$  & $r \geq 3$  & $r \geq 2$  & $r \geq 3$  \\
$(4, 5)$   & $r \geq 4$  & $r \geq 4$  & $r \geq 3$  & $r \geq 4$  \\
$(3, 5)$   & $r \geq 4$  & $r \geq 4$  & $r \geq 3$  & $r \geq 3$  \\
$(5, 7)$   & $r \geq 6$  & $r \geq 6$  & {$r \geq 6$}  & $r \geq 6$  \\
$(5, 9)$   & $r \geq 9$  & $r \geq 9$  & {$r \geq 8$}  & $r \geq 8$  \\
$(9, 11)$  & $r \geq 12$ & $r \geq 12$ & $r \geq 11$ & $r \geq 11$ \\
$(11, 13)$ & $r \geq 15$ & $r \geq 15$ & $r \geq 14$ & $r \geq 14$ \\
$(11, 15)$ & $r \geq 17$ & $r \geq 17$ & {$r \geq 16$} & $r \geq 16$ \\
\hline
\end{tabular}
\end{table}
The linear programming bound is strictly stronger than the Plotkin bound at $(3, 3)$, $(3, 4)$ and
$(4, 4)$ under both assignments, at $(2, 3)$ under the first assignment, and at $(4, 5)$ under the
second. At the remaining pairs, the two bounds agree. The bound is larger for the first assignment at
$(2, 3)$, $(3, 5)$, $(5, 9)$, $(9, 11)$, $(11, 13)$ and $(11, 15)$, and the two assignments agree at
the rest. At $(3, 3)$ and $(4, 4)$ they must agree, since the two assignments describe the same code
when $d_1 = d_2$. Protecting $\mathcal{P}_{g_2}$, which has the larger effective number of blocks,
at the larger distance requirement is the more demanding of the two, and the linear program detects
this even though the pair of distance requirements is the same in both columns.
\end{example}

\subsection{Linear Programming Bound for Multiple Partitions}
\label{subsec:lp-multi}

We now extend the bound to $H$ partitions $\mathcal{P}_1, \dots, \mathcal{P}_H$ of $\mathbb{F}_q^k$ with distance requirements $d_1 \le d_2 \le \dots \le d_H$. The distance distribution is split according to the set of partitions that separate a given pair. For the image $C$ of a systematic $(\mathcal{P}_h : d_h;\ h \in [H])$-FCPC, each subset $S \subseteq [H]$, and each $i \in \{0, \dots, n\}$, define
\begin{equation}
B_i^{S} = \frac{1}{M}\,\bigl|\{(x, y) \in C^2 : d(x, y) = i,\ \mathcal{P}_j(x) \ne \mathcal{P}_j(y) \ \forall\, j \in S,\ \mathcal{P}_j(x) = \mathcal{P}_j(y) \ \forall\, j \in [H] \setminus S\}\bigr|.
\label{eq:multi-BS}
\end{equation}
The superscript represents the set of partitions in which the pair differs, as in Section~\ref{subsec:lp-two}. The $2^H$ distributions are disjoint and sum to the full distance distribution,
\begin{equation}
B_i = \sum_{S \subseteq [H]} B_i^{S}.
\label{eq:multi-full}
\end{equation}

For $T \subseteq [H]$, write $\mathcal{P}_T = \bigvee_{j \in T} \mathcal{P}_j$, with the convention that $\mathcal{P}_{\emptyset}$ is the trivial one-block partition, so $\widetilde{E}(\mathcal{P}_{\emptyset}) = 1$. Two codewords lie in the same block of $\mathcal{P}_T$ if and only if they lie in the same block of $\mathcal{P}_j$ for every $j \in T$. The same-block distance distribution of $C$ with respect to $\mathcal{P}_T$ is therefore
\begin{equation}
\sum_{\substack{S \subseteq [H] \\ S \cap T = \emptyset}} B_i^{S}.
\label{eq:multi-sameblock}
\end{equation}
For $T = \emptyset$ this is the full distribution $B_i$. Lemma~\ref{lem:macwilliams-block} applies to each of these $2^H$ combinations and shows that its Krawtchouk transform is nonnegative. As in the two-partition case, the positivity holds only for the same-block combinations \eqref{eq:multi-sameblock}.

A pair counted by $B_i^{S}$ with $S \ne \emptyset$ differs in $\mathcal{P}_j$ exactly for $j \in S$, so it must be separated by distance $d_j$ for every $j \in S$. Since $d_1 \le \dots \le d_H$, the combined requirement is $d_{\max S}$, where $\max S$ denotes the largest element of $S$.

Fix a length $n$ and distances $d_1 \le \dots \le d_H$. Consider the linear program in the $2^H(n+1)$ variables $(B_i^{S})_{i=0,\dots,n}^{S \subseteq [H]}$:
\begin{linearprogram}\label{lp:multi}
\begin{equation}
M_{\mathrm{LP}}\bigl(n, d_1, \dots, d_H, \{\widetilde{E}(\mathcal{P}_T)\}_{\emptyset \ne T \subseteq [H]}\bigr) = \max \ \sum_{i=0}^{n} \sum_{S \subseteq [H]} B_i^{S}
\label{eq:lp-multi-obj}
\end{equation}
subject to
\begin{align}
&B_0^{\emptyset} = 1, \quad B_0^{S} = 0, && \emptyset \ne S \subseteq [H], \label{eq:mlp1}\\
&B_i^{S} \ge 0, && i = 0, \dots, n, \quad S \subseteq [H], \label{eq:mlp2}\\
&B_i^{S} = 0, && 1 \le i \le d_{\max S} - 1, \quad \emptyset \ne S \subseteq [H], \label{eq:mlp3}\\
&\sum_{i=0}^{n} \Biggl(\ \sum_{\substack{S \subseteq [H] \\ S \cap T = \emptyset}} B_i^{S}\Biggr) K_j(i) \ge 0, && j = 0, \dots, n, \quad T \subseteq [H], \label{eq:mlp4}\\
&\sum_{i=0}^{n} B_i = \widetilde{E}(\mathcal{P}_T) \sum_{i=0}^{n} \sum_{\substack{S \subseteq [H] \\ S \cap T = \emptyset}} B_i^{S}, && \emptyset \ne T \subseteq [H]. \label{eq:mlp5}
\end{align}
\end{linearprogram}
Constraint \eqref{eq:mlp3} collects the distance constraints. Constraint \eqref{eq:mlp4} collects the Delsarte inequalities for the $2^H$ same-block distributions \eqref{eq:multi-sameblock}, where $T = \emptyset$ gives the whole-code inequality. Constraint \eqref{eq:mlp5} is the partition condition \eqref{eq:partition-cond} applied to every join $\mathcal{P}_T$, and the common left-hand side is $\sum_i B_i = M$.

\begin{theorem}
\label{thm:lp-multi}
Let $\mathcal{P}_1, \dots, \mathcal{P}_H$ be partitions of $\mathbb{F}_q^k$ with distance requirements $d_1 \le \dots \le d_H$, and let $\bm{d} = (d_1, \dots, d_H)$. Then
\[
r_{\mathcal{P}_1, \dots, \mathcal{P}_H}(k : \bm{d}) \ge \min\bigl\{\, n' : M_{\mathrm{LP}}\bigl(n', d_1, \dots, d_H, \{\widetilde{E}(\mathcal{P}_T)\}_{\emptyset \ne T \subseteq [H]}\bigr) \ge q^k \,\bigr\} - k.
\]
\end{theorem}

\begin{proof}
Let $\mathcal{C} : \mathbb{F}_q^k \to \mathbb{F}_q^{k + r}$ be an optimal $(\mathcal{P}_h : d_h;\ h \in [H])$-FCPC, so $r = r_{\mathcal{P}_1, \dots, \mathcal{P}_H}(k : \bm{d})$ and $n = k + r$, and let $C = \mathrm{Im}(\mathcal{C})$. The distance distributions \eqref{eq:multi-BS} of $C$ satisfy \eqref{eq:mlp1}--\eqref{eq:mlp5}. The distance constraints \eqref{eq:mlp3} hold because a pair counted by $B_i^{S}$ differs in $\mathcal{P}_j$ for every $j \in S$ and is therefore separated by distance $d_{\max S}$. The Delsarte inequalities \eqref{eq:mlp4} hold by Lemma~\ref{lem:macwilliams-block} applied to the partition of $C$ induced by $\mathcal{P}_T$ for each $T \subseteq [H]$. The equalities \eqref{eq:mlp5} hold by the partition condition \eqref{eq:partition-cond} applied to each join $\mathcal{P}_T, \ T \subseteq [H]$. Since the distance distributions of $C$ are feasible for the program, $M_{\mathrm{LP}}\bigl(n, d_1, \dots, d_H, \{\widetilde{E}(\mathcal{P}_T)\}\bigr) \ge \sum_i B_i = q^k$, and the bound follows as in the proof of Theorem~\ref{thm:lp-single}.
\end{proof}

\begin{remark}
\label{rem:lp-multi-recover}
For $H = 2$, LP\ref{lp:multi} is exactly LP\ref{lp:two}. For $H = 1$ the variables are $B_i^{\emptyset}$ and $B_i^{\{1\}}$, and the change of variables $B_i^{\{1\}} = B_i - B_i^{\emptyset}$ identifies the program with the single-partition program LP\ref{lp:single}.
\end{remark}

\begin{remark}
\label{rem:lp-multi-subfamily}
The $H$-partition program is at least as strong as the program for any subfamily of the partitions. Let $T' \subseteq [H]$ and consider LP\ref{lp:multi} for the partitions $\{\mathcal{P}_j\}_{j \in T'}$. Given a feasible point of the full program, set $\widetilde{B}_i^{S'} = \sum_{S : S \cap T' = S'} B_i^{S}$ for $S' \subseteq T'$. For every $T \subseteq T'$ this aggregation satisfies
\[
\sum_{\substack{S' \subseteq T' \\ S' \cap T = \emptyset}} \widetilde{B}_i^{S'} \;=\; \sum_{\substack{S \subseteq [H] \\ S \cap T = \emptyset}} B_i^{S},
\]
so the Delsarte inequalities \eqref{eq:mlp4} and the partition conditions \eqref{eq:mlp5} of the subfamily program are exactly those of the full program indexed by $T \subseteq T'$. Conditions \eqref{eq:mlp1} and \eqref{eq:mlp2} are preserved by the aggregation, and the distance constraints \eqref{eq:mlp3} hold because $\max S \ge \max S'$ whenever $S \cap T' = S'$. The objective value is unchanged, since $\sum_{S' \subseteq T'} \widetilde{B}_i^{S'} = B_i$. Hence the maximum of the full program is at most that of the subfamily program, and the resulting redundancy bound is at least as large.
\end{remark}

\begin{example}\label{ex:lp-three}
Let $u = (u_1, u_2, u_3, u_4,u_5) \in \mathbb{F}_3^5$ and consider the functions  $g_1 : \mathbb{F}_3^5 \to \mathbb{F}_3$, $g_2 : \mathbb{F}_3^5 \to \mathbb{F}_3^2$, and
$g_3 : \mathbb{F}_3^5 \to \mathbb{F}_3^3$ with
\[
g_1(u) = u_1 u_2 + u_3 u_4 \pmod 3, \qquad
g_2(u) = (u_1 + u_2,\ u_4 + u_5), \qquad
g_3(u) = (u_1, u_2, u_3).
\]
 Write $\mathcal{P}_1$, $\mathcal{P}_2$, and
$\mathcal{P}_3$ for the domain partitions induced by $g_1, g_2$ and $g_3$, respectively. The function $g_1$ takes the value $0$
on $99$ of the $243$ points and each of the values $1$ and $2$ on $72$ points, so $\mathcal{P}_1$ has three blocks of sizes $99$, $72$ and $72$. The blocks of
$\mathcal{P}_2$ all have size $27$, and those of $\mathcal{P}_3$ all have size $9$.
Table~\ref{tab:block_sizes} lists the block sizes of these partitions and of their joins, together
with the effective number of blocks of each.

\begin{table}[!t]
\centering
\caption{Block sizes and effective number of blocks of the partitions induced by $g_1$, $g_2$, and
$g_3$ and of their joins, in Example~\ref{ex:lp-three}.}
\label{tab:block_sizes}
\renewcommand{\arraystretch}{1.15}
\setlength{\tabcolsep}{9pt}
\begin{tabular}{c|c|c}
\hline
Partition $\mathcal{P}_T$ & Block sizes & $\widetilde{E}(\mathcal{P}_T)$ \\
\hline
$\mathcal{P}_1$ & $(99, 72,72)$ & $243/83$ \\
$\mathcal{P}_2$ & $(\underbrace{27, \ldots, 27}_{9})$ & $9$ \\
$\mathcal{P}_3$ & $(\underbrace{9, \ldots, 9}_{27})$ & $27$ \\
$\mathcal{P}_1 \vee \mathcal{P}_2$ & $(\underbrace{6, \ldots, 6}_{9},\underbrace{9, \ldots, 9}_{9},\underbrace{12, \ldots, 12}_{9})$ & $729/29$ \\
$\mathcal{P}_1 \vee \mathcal{P}_3$ & $(\underbrace{3, \ldots, 3}_{54}, \underbrace{9, \ldots, 9}_{9})$ & $243/5$ \\
$\mathcal{P}_2 \vee \mathcal{P}_3$ & $(\underbrace{3, \ldots, 3}_{81})$ & $81$ \\
$\mathcal{P}_1 \vee \mathcal{P}_2 \vee \mathcal{P}_3$
  & $(\underbrace{3, \ldots, 3}_{27}, \underbrace{1, \ldots, 1}_{162})$ & $729/5$ \\
\hline
\end{tabular}
\end{table}

Solving LP\ref{lp:multi} gives the bounds in
Table~\ref{tab:three-bounds}, for twelve triples $(d_1, d_2, d_3)$.

\begin{table}[!t]
\centering
\caption{Lower bounds on the optimal redundancy in Example~\ref{ex:lp-three}.}
\label{tab:three-bounds}
\renewcommand{\arraystretch}{1.15}
\setlength{\tabcolsep}{8pt}
\begin{tabular}{c|cc||c|cc}
\hline
$(d_1, d_2, d_3)$ & Plotkin \cite{LBZY2023} & LP
& $(d_1, d_2, d_3)$ & Plotkin \cite{LBZY2023} & LP \\
\hline
$(3, 5, 7)$     & $r \geq 6$  & $r \geq 7$  & $(3, 7, 13)$  & $r \geq 15$ & $r \geq 15$ \\
$(5, 7, 9)$     & $r \geq 9$  & $r \geq 10$ & $(3, 7, 15)$  & $r \geq 17$ & $r \geq 18$ \\
$(7, 9, 11)$    & $r \geq 12$ & $r \geq 13$ & $(3, 9, 11)$  & $r \geq 12$ & $r \geq 13$ \\
$(9, 11, 13)$   & $r \geq 15$ & $r \geq 16$ & $(5, 9, 17)$  & $r \geq 20$ & $r \geq 21$ \\
$(11, 13, 15)$  & $r \geq 18$ & $r \geq 18$ & $(5, 11, 19)$ & $r \geq 23$ & $r \geq 24$ \\
$(13, 15, 17)$  & $r \geq 21$ & $r \geq 21$ & $(5, 13, 17)$ & $r \geq 21$ & $r \geq 21$ \\
\hline
\end{tabular}
\end{table}

The linear programming bound improves on the Plotkin bound by one symbol of redundancy at eight of
the twelve triples. The two agree at $(11, 13, 15)$, $(13, 15, 17)$, $(3, 7, 13)$ and
$(5, 13, 17)$. The first six triples have consecutive distance
requirements, and the remaining six spread the requirements further apart.
\end{example}


\section{Multi-Step Construction Method}\label{sec:construction}

We now present a multi-step construction method for generalized FCPCs. The construction method produces a valid encoding and therefore gives an upper bound on the optimal redundancy.

\subsection{Construction Method}\label{subsec:construction-proc}

Let $\mathcal{P}_1, \mathcal{P}_2, \ldots, \mathcal{P}_H$ be partitions of $\mathbb{F}_q^k$ with distance requirements $d_1, d_2, \ldots, d_H$, respectively. Without loss of generality, assume
\[
d_1 \leq d_2 \leq \cdots \leq d_H.
\]
The construction method proceeds in $H$ steps. Let $r_h$ denote the redundancy added at Step~$h$, for $h \in [H]$.

\begin{itemize}
    \item \textbf{Step~1} (Baseline protection for all partitions with distance $d_1$). Define $\mathcal{Q}_1 = \mathcal{P}_1 \vee \mathcal{P}_2 \vee \cdots \vee \mathcal{P}_H$. Construct a systematic encoding $\mathcal{C}_1 : \mathbb{F}_q^k \to \mathbb{F}_q^{k+r_1}$ such that for all $u, v \in \mathbb{F}_q^k$ lying in different blocks of $\mathcal{Q}_1$,
    \[
    d\big(\mathcal{C}_1(u),\, \mathcal{C}_1(v)\big) \geq d_1.
    \]

    \item \textbf{Step~$h$}, for $2 \leq h \leq H$ (Upgrade protection from $d_{h-1}$ to $d_h$ for $\mathcal{P}_h, \ldots, \mathcal{P}_H$). Define the partition $\mathcal{Q}_h = \mathcal{P}_h \vee \mathcal{P}_{h+1} \vee \cdots \vee \mathcal{P}_H$. Treat the image $\mathrm{Im}(\mathcal{C}_{h-1}) \subseteq \mathbb{F}_q^{k + \sum_{i=1}^{h-1} r_i}$ as the new message set, and construct a systematic map
    \[
    \varphi_h : \mathrm{Im}(\mathcal{C}_{h-1}) \to \mathbb{F}_q^{k + \sum_{i=1}^{h} r_i},
    \]
    by which we mean that $\varphi_h$ preserves its input as a prefix, so that $\varphi_h(x) = \big(x,\, \psi_h(x)\big)$ for some map $\psi_h : \mathrm{Im}(\mathcal{C}_{h-1}) \to \mathbb{F}_q^{r_h}$, such that for all $u, v \in \mathbb{F}_q^k$ lying in different blocks of $\mathcal{Q}_h$,
    \[
    d\big(\varphi_h(\mathcal{C}_{h-1}(u)),\, \varphi_h(\mathcal{C}_{h-1}(v))\big) \geq d_h.
    \]
    Then define $\mathcal{C}_h = \varphi_h \circ \mathcal{C}_{h-1}$.
\end{itemize}

The partitions $\mathcal{Q}_1, \ldots, \mathcal{Q}_H$ are the same joins that appear in the lower bound of Theorem~\ref{thm:lower}.

The overall encoder after Step~$H$ is
\[
\mathcal{C}_H : \mathbb{F}_q^k \to \mathbb{F}_q^{k + \sum_{i=1}^{H} r_i}, \qquad \mathcal{C}_H = \varphi_H \circ \varphi_{H-1} \circ \cdots \circ \varphi_2 \circ \mathcal{C}_1.
\]
Then $\mathcal{C}_H$ is a $(\mathcal{P}_h : d_h;\, h \in [H])$-FCPC with total redundancy $r = \sum_{i=1}^{H} r_i$.

\subsection{Proof of Correctness}\label{subsec:construction-proof}

\begin{theorem}\label{thm:construction}
The encoder $\mathcal{C}_H$ produced by the $H$-step construction method is a valid $(\mathcal{P}_h : d_h;\, h \in [H])$-FCPC.
\end{theorem}

\begin{proof}
Fix any $h \in [H]$ and let $u, v \in \mathbb{F}_q^k$ lie in different blocks of $\mathcal{P}_h$. Since $\mathcal{Q}_h = \mathcal{P}_h \vee \mathcal{P}_{h+1} \vee \cdots \vee \mathcal{P}_H$ is a refinement of $\mathcal{P}_h$, the vectors $u$ and $v$ also lie in different blocks of $\mathcal{Q}_h$. By the design of Step~$h$, we have
\[
d\big(\mathcal{C}_h(u),\, \mathcal{C}_h(v)\big) \geq d_h.
\]

It remains to show that the subsequent steps do not decrease this distance. For any $\ell > h$, the map $\varphi_\ell$ preserves its input as a prefix, so for any two vectors $x, y \in \mathrm{Im}(\mathcal{C}_{\ell-1})$,
\[
d\big(\varphi_\ell(x),\, \varphi_\ell(y)\big) \geq d(x, y).
\]
Applying this inequality iteratively for $\ell = h+1, h+2, \ldots, H$, we obtain
\[
d\big(\mathcal{C}_H(u),\, \mathcal{C}_H(v)\big) \geq d\big(\mathcal{C}_h(u),\, \mathcal{C}_h(v)\big) \geq d_h.
\]

Since this holds for every $h \in [H]$ and every pair $u, v$ in different blocks of $\mathcal{P}_h$, the encoder $\mathcal{C}_H$ is a valid $(\mathcal{P}_h : d_h;\, h \in [H])$-FCPC.
\end{proof}

\begin{remark}\label{rem:rms}
By Theorem~\ref{thm:construction}, every encoding produced by the $H$-step construction is a valid
$(\mathcal{P}_h : d_h;\, h \in [H])$-FCPC. Hence
\[
r_{\mathcal{P}_1, \ldots, \mathcal{P}_H}(k : \bm{d}) \leq r^{\mathrm{ms}}_{\mathcal{P}_1, \ldots, \mathcal{P}_H}(k : \bm{d}),
\]
where $r^{\mathrm{ms}}_{\mathcal{P}_1, \ldots, \mathcal{P}_H}(k : \bm{d})$ is as defined in Subsection~\ref{subsec:summary}. The multi-step redundancy thus provides an upper bound on the
optimal redundancy.
\end{remark}


\subsection{Examples}\label{subsec:examples}

We now illustrate the multi-step construction method through two examples over $\mathbb{F}_3^3$. The first example demonstrates that the construction method can achieve strictly smaller redundancy than the sum of the redundancies of separate FCPCs designed for each partition individually. The second example shows that even a single FCPC designed for the join partition $\mathcal{P}_1 \vee \mathcal{P}_2$ at the larger distance requirement $d_2$ can require more redundancy than the construction method, since the construction method protects each partition only at the level that partition requires.

Together, these examples show why the proposed construction method saves redundancy by building the encoding in multiple steps. Each step reuses the parity symbols added in the earlier steps and upgrades protection only for the partitions that require a stronger distance, rather than treating the partitions separately or protecting all of them at the strongest distance.

\begin{example}\label{ex:ms-sum}
Let $f_1 : \mathbb{F}_3^3 \to \{0, 1, 2, 3\}$ be the Hamming weight function defined by
\[
f_1(u) = \mathrm{wt}(u), \quad \text{for } u \in \mathbb{F}_3^3.
\]
The domain partition induced by $f_1$ is
\begin{align*}
\mathcal{P}_1 = \big\{
& f_1^{-1}(0) = \{000\},\;
  f_1^{-1}(1) = \{100, 010, 001, 200, 020, 002\}, \\
& f_1^{-1}(2) = \{110, 101, 011, 220, 202, 022, 120, 102, 012, 210, 201, 021\}, \\
& f_1^{-1}(3) = \{111, 222, 112, 121, 211, 221, 122, 212\} \big\}.
\end{align*}
Let $f_2 : \mathbb{F}_3^3 \to \mathbb{F}_3$ be the function defined by
\[
f_2(x_1, x_2, x_3) = x_1 + x_2 + x_3 \pmod{3}.
\]
The domain partition induced by $f_2$ is
\begin{align*}
\mathcal{P}_2 = \big\{
& f_2^{-1}(0) = \{000, 012, 021, 102, 120, 201, 210, 111, 222\}, \\
& f_2^{-1}(1) = \{001, 010, 100, 022, 202, 220, 112, 121, 211\}, \\
& f_2^{-1}(2) = \{002, 020, 200, 011, 101, 110, 122, 212, 221\} \big\}.
\end{align*}

We set the distance requirements $d_1 = 3$ and $d_2 = 5$ (corresponding to error-correction capabilities $t_1 = 1$ and $t_2 = 2$, respectively).
Table~\ref{tab:ex1_P1} shows an optimal $(\mathcal{P}_1 : 3)$-encoding with redundancy $r_{\mathcal{P}_1}(k: 3) = 2$, and Table~\ref{tab:ex1_P2} shows an optimal $(\mathcal{P}_2 : 5)$-encoding with redundancy $r_{\mathcal{P}_2}(k:5) = 4$.

\begin{table}[ht]
\centering
\caption{$(\mathcal{P}_1: 3)$-encoding with redundancy $2$.}
\label{tab:ex1_P1}
\begin{tabular}{c|c}
\hline
Block of $\mathcal{P}_1$ & Redundancy \\
\hline
$f_1^{-1}(0)$ & $00$ \\
$f_1^{-1}(1)$ & $11$ \\
$f_1^{-1}(2)$ & $20$ \\
$f_1^{-1}(3)$ & $01$ \\
\hline
\end{tabular}
\end{table}

\begin{table}[ht]
\centering
\caption{$(\mathcal{P}_2: 5)$-encoding with redundancy $4$.}
\label{tab:ex1_P2}
\begin{tabular}{c|c}
\hline
Block of $\mathcal{P}_2$ & Redundancy \\
\hline
$f_2^{-1}(0)$ & $0000$ \\
$f_2^{-1}(1)$ & $1111$ \\
$f_2^{-1}(2)$ & $2222$ \\
\hline
\end{tabular}
\end{table}

The join partition is
\begin{align*}
\mathcal{P}_1 \vee \mathcal{P}_2 = \big\{
& \{000\},\, \{001, 010, 100\},\, \{002, 020, 200\}, \\
& \{011, 101, 110\},\, \{012, 021, 102, 120, 201, 210\},\, \{022, 202, 220\}, \\
& \{111, 222\},\, \{112, 121, 211\},\, \{122, 212, 221\} \big\}.
\end{align*}

We now apply the multi-step construction method. In Step~1, we construct a $(\mathcal{P}_1 \vee \mathcal{P}_2 : 3)$-encoding with redundancy $r_1 = 3$. In Step~2, we upgrade the distance for $\mathcal{P}_2$ from $d_1 = 3$ to $d_2 = 5$ by appending additional redundancy $r_2 = 2$. The upgrade assigns one Step~2 parity vector to each block of $\mathcal{P}_2$, by the assignment $f_2^{-1}(0) \mapsto 00$, $f_2^{-1}(1) \mapsto 11$, $f_2^{-1}(2) \mapsto 22$. Table~\ref{tab:ex1_joint} shows the resulting encoding.

\begin{table}[ht]
\centering
\caption{Multi-step construction of a $(\mathcal{P}_h : d_h;\, h \in [2])$-FCPC with $d_1 = 3$ and $d_2 = 5$ in Example~\ref{ex:ms-sum}.}
\label{tab:ex1_joint}
\begin{tabular}{c|c|c}
\hline
Block of $\mathcal{P}_1 \vee \mathcal{P}_2$ & Redundancy (Step~1) & Redundancy (Step~2) \\
\hline
$\{000\}$ & $000$ & $00$ \\
$\{001, 010, 100\}$ & $110$ & $11$ \\
$\{002, 020, 200\}$ & $220$ & $22$ \\
$\{011, 101, 110\}$ & $200$ & $22$ \\
$\{012, 021, 102, 120, 201, 210\}$ & $001$ & $00$ \\
$\{022, 202, 220\}$ & $100$ & $11$ \\
$\{111, 222\}$ & $010$ & $00$ \\
$\{112, 121, 211\}$ & $120$ & $11$ \\
$\{122, 212, 221\}$ & $111$ & $22$ \\
\hline
\end{tabular}
\end{table}

In this example, the redundancy achieved by the multi-step construction method is $5$, therefore
\[
r^{\mathrm{ms}}_{\mathcal{P}_1, \mathcal{P}_2}(k: (3, 5)) \leq 5,
\]
while the sum of the individual optimal redundancies is
\[
r_{\mathcal{P}_1}(k:3) + r_{\mathcal{P}_2}(k:5) = 2 + 4 = 6.
\]
Thus,
\[
r_{\mathcal{P}_1, \mathcal{P}_2}(k : (3, 5)) \leq r^{\mathrm{ms}}_{\mathcal{P}_1, \mathcal{P}_2}(k: (3, 5)) \leq 5 < 6 = r_{\mathcal{P}_1}(k: 3) + r_{\mathcal{P}_2}(k:5),
\]
demonstrating that the multi-step construction method achieves a strict redundancy saving over using separate codes for each partition.
\end{example}

\begin{example}\label{ex:ms-join}
Let $g_1 : \mathbb{F}_3^3 \to \{0, 1, 2, 3\}$ be the Hamming weight function defined by
\[
g_1(u) = \mathrm{wt}(u), \quad \text{for } u \in \mathbb{F}_3^3.
\]
The domain partition induced by $g_1$ is
\begin{align*}
\mathcal{P}_1 = \big\{
& g_1^{-1}(0) = \{000\},\;
  g_1^{-1}(1) = \{100, 010, 001, 200, 020, 002\}, \\
& g_1^{-1}(2) = \{110, 101, 011, 220, 202, 022, 120, 102, 012, 210, 201, 021\}, \\
& g_1^{-1}(3) = \{111, 222, 112, 121, 211, 221, 122, 212\} \big\},
\end{align*}
which is the same as the partition $\mathcal{P}_1$ in Example~\ref{ex:ms-sum}.

Let $g_2 : \mathbb{F}_3^3 \to \mathbb{F}_3$ be the coordinate projection defined by
\[
g_2(x_1, x_2, x_3) = x_1.
\]
The domain partition induced by $g_2$ is
\begin{align*}
\mathcal{P}_2 = \big\{
& g_2^{-1}(0) = \{000,001,002,010,020,011,022,012,021\}, \\
& g_2^{-1}(1) = \{100,101,102,110,120,111,122,112,121\}, \\
& g_2^{-1}(2) = \{200,201,202,210,220,211,222,212,221\} \big\}.
\end{align*}

We set $d_1 = 3$ and $d_2 = 5$. Since $\mathcal{P}_1$ is the same as in Example~\ref{ex:ms-sum}, we have $r_{\mathcal{P}_1}(k:3) = 2$ (see Table~\ref{tab:ex1_P1}). We also have $r_{\mathcal{P}_2}(k:5) = 4$. The messages $000$ and $100$ lie in different blocks of $\mathcal{P}_2$ at distance $1$, which forces a parity distance of at least $4$, and the assignment $g_2^{-1}(0) \mapsto 0000$, $g_2^{-1}(1) \mapsto 1111$, $g_2^{-1}(2) \mapsto 2222$ attains it.

The join partition $\mathcal{P}_1 \vee \mathcal{P}_2$ has $9$ blocks:
\begin{align*}
\mathcal{P}_1 \vee \mathcal{P}_2 = \big\{
& \{000\},\, \{001, 002, 010, 020\},\, \{011, 012, 021, 022\}, \\
& \{100\},\, \{101, 102, 110, 120\},\, \{111, 112, 121, 122\}, \\
& \{200\},\, \{201, 202, 210, 220\},\, \{211, 212, 221, 222\} \big\}.
\end{align*}

We now apply the multi-step construction method. In Step~1, we construct a $(\mathcal{P}_1 \vee \mathcal{P}_2 : 3)$-encoding with redundancy $r_1 = 2$. In Step~2, we upgrade the distance for $\mathcal{P}_2$ from $d_1 = 3$ to $d_2 = 5$ by appending additional redundancy $r_2 = 2$, using the assignment $g_2^{-1}(0) \mapsto 00$, $g_2^{-1}(1) \mapsto 11$, $g_2^{-1}(2) \mapsto 22$. Table~\ref{tab:ex2_joint} shows the resulting encoding.

\begin{table}[ht]
\centering
\caption{Multi-step construction of a $(\mathcal{P}_h : d_h;\, h \in [2])$-FCPC with $d_1 = 3$ and $d_2 = 5$ in Example~\ref{ex:ms-join}.}
\label{tab:ex2_joint}
\begin{tabular}{c|c|c}
\hline
Block of $\mathcal{P}_1 \vee \mathcal{P}_2$ & Redundancy (Step~1) & Redundancy (Step~2) \\
\hline
$\{000\}$ & $00$ & $00$ \\
$\{001, 002, 010, 020\}$ & $11$ & $00$ \\
$\{011, 012, 021, 022\}$ & $22$ & $00$ \\
$\{100\}$ & $21$ & $11$ \\
$\{101, 102, 110, 120\}$ & $02$ & $11$ \\
$\{111, 112, 121, 122\}$ & $10$ & $11$ \\
$\{200\}$ & $12$ & $22$ \\
$\{201, 202, 210, 220\}$ & $20$ & $22$ \\
$\{211, 212, 221, 222\}$ & $01$ & $22$ \\
\hline
\end{tabular}
\end{table}

In this example, the redundancy achieved by the multi-step construction method is $4$, therefore
\[
r^{\mathrm{ms}}_{\mathcal{P}_1, \mathcal{P}_2}(k: (3, 5)) \leq 4.
\]
We now show that the optimal redundancy of a single FCPC for the join partition $\mathcal{P}_1 \vee \mathcal{P}_2$ at the larger distance requirement $d_2 = 5$ satisfies $r_{\mathcal{P}_1 \vee \mathcal{P}_2}(k : 5) = 5$, which is strictly larger than $4$.

Consider four message vectors $u_1=000, u_2=100, u_3=200, u_4=010$, which lie in four distinct blocks of $\mathcal{P}_1 \vee \mathcal{P}_2$, with $d(u_1,u_2)=d(u_1,u_3)=d(u_1,u_4)=d(u_2,u_3)=1$ and $d(u_2,u_4)=d(u_3,u_4)=2$. Let $p_1, p_2, p_3, p_4$ denote the parity vectors assigned to the messages $u_1, u_2, u_3, u_4$, respectively, in a $(\mathcal{P}_1 \vee \mathcal{P}_2 : 5)$-FCPC of redundancy $r$. The FCPC condition $d(\mathcal{C}(u_i),\mathcal{C}(u_j)) \geq 5$ together with the message distances above forces
$$d(p_1, p_2), d(p_1,p_3), d(p_1, p_4), d(p_2,p_3) \geq 4, \quad \text{and} \quad d(p_2, p_4), d(p_3,p_4) \geq 3.$$
The first group of inequalities already gives $r \geq 4$. Suppose $r = 4$. Without loss of generality take $p_1 = 0000$. Each of $p_2, p_3, p_4$ then differs from $p_1$ in all four coordinates, so $p_2, p_3, p_4 \in \{1,2\}^4$. Since $d(p_2, p_3) \geq 4$, the vectors $p_2$ and $p_3$ differ in every coordinate, so $p_3$ is obtained from $p_2$ by exchanging the symbols $1$ and $2$ in each coordinate. Consequently $p_4$ agrees with exactly one of $p_2$ and $p_3$ in each coordinate, which gives $d(p_2, p_4) + d(p_3, p_4) = 4$. This contradicts $d(p_2, p_4), d(p_3,p_4) \geq 3$. Therefore $r_{\mathcal{P}_1 \vee \mathcal{P}_2}(k:5)\geq 5$.

A $(\mathcal{P}_1 \vee \mathcal{P}_2 : 5)$-FCPC with redundancy $5$ is given by the assignment in Table~\ref{tab:ex2_join5}.

\begin{table}[ht]
\centering
\caption{$(\mathcal{P}_1 \vee \mathcal{P}_2: 5)$-encoding with redundancy $5$.}
\label{tab:ex2_join5}
\begin{tabular}{c|c}
\hline
Block of $\mathcal{P}_1 \vee \mathcal{P}_2$ & Redundancy \\
\hline
$\{000\}$ & $00000$ \\
$\{001, 002, 010, 020\}$ & $11110$ \\
$\{011, 012, 021, 022\}$ & $22200$ \\
$\{100\}$ & $22120$ \\
$\{101, 102, 110, 120\}$ & $21001$ \\
$\{111, 112, 121, 122\}$ & $10020$ \\
$\{200\}$ & $12201$ \\
$\{201, 202, 210, 220\}$ & $00121$ \\
$\{211, 212, 221, 222\}$ & $01012$ \\
\hline
\end{tabular}
\end{table}

Therefore, we conclude that $r_{\mathcal{P}_1 \vee \mathcal{P}_2}(k : 5) = 5$. Furthermore, the optimal redundancy of the GFCPC satisfies
\[
r_{\mathcal{P}_1, \mathcal{P}_2}(k : (3, 5)) \geq r_{\mathcal{P}_2}(k : 5)= 4,
\]
which matches the upper bound, so $r_{\mathcal{P}_1, \mathcal{P}_2}(k : (3, 5)) = 4$. The relevant quantities compare as follows:
\[
r_{\mathcal{P}_1, \mathcal{P}_2}(k: (3,5)) = r^{\mathrm{ms}}_{\mathcal{P}_1, \mathcal{P}_2}(k: (3,5)) \leq 4 < r_{\mathcal{P}_1 \vee \mathcal{P}_2}(k: 5) = 5 < r_{\mathcal{P}_1}(k:3) + r_{\mathcal{P}_2}(k: 5) = 6.
\]
This demonstrates that the multi-step construction method achieves strictly smaller redundancy not only compared to using separate codes for each partition, but also compared to using a single FCPC for the join partition at the larger distance requirement. In this case, the multi-step construction method is optimal. Since the two partitions of $[2]$ give the values $6$ and $5$, the upper bound of Theorem~\ref{thm:upper} evaluates to $5$, which the multi-step construction method strictly improves. The saving arises because the multi-step construction method does not force the partition $\mathcal{P}_1$ to be protected at the stronger level $t_2 = 2$. It needs only $t_1 = 1$.
\end{example}


\section{Conclusion}\label{sec:conclusion}
We introduced generalized function-correcting partition codes (FCPCs), a framework that simultaneously protects multiple partitions of the message space with different distance requirements. This framework unifies and extends both FCCs with data protection and FCPCs for multiple functions with a common error-correction level. We introduced the distance requirement matrix for generalized FCPCs, used it to characterize the optimal redundancy in terms of $\mathcal{D}$-codes, and obtained computable lower bounds from its submatrices. We derived general lower and upper bounds on the optimal redundancy. For the binary case with two partitions, we established improved lower bounds under specific structural conditions on the partitions. We then derived a linear programming lower bound on the optimal redundancy of GFCPCs over $\mathbb{F}_q$. This bound applies to generalized FCPCs, and hence to FCPCs, to FCCs, and to FCCs with data protection. We presented a multi-step construction method for these codes and proved its correctness. Examples show that the construction method can require strictly less redundancy than protecting each partition separately and than protecting the join partition at the largest distance requirement, and that it is optimal in some cases.

Several directions remain open. Determining the exact optimal redundancy for specific families of partitions and closing the gap between the lower and upper bounds in general are natural next steps. Example~\ref{ex:q3} shows that the improved binary bounds do not extend as stated to $q \geq 3$, so finding analogues for non-binary alphabets is a further question. 



\begin{thebibliography}{99}

\bibitem{LBZY2023}
A.~Lenz, R.~Bitar, A.~Wachter-Zeh, and E.~Yaakobi,
``Function-correcting codes,''
\textit{IEEE Trans. Inf. Theory}, vol.~69, no.~9, pp.~5604--5618, Sep.~2023.

\bibitem{RRHH20261}
C.~Rajput, B.~S.~Rajan, R.~Freij-Hollanti, and C.~Hollanti,
``Function-correcting partition codes,''
\textit{arXiv preprint arXiv:2601.06450}, Jan.~2026.

\bibitem{RRHH20251}
C.~Rajput, B.~S.~Rajan, R.~Freij-Hollanti, and C.~Hollanti,
``Function-correcting codes with data protection,''
\textit{IEEE Trans. Inf. Theory}, vol.~72, no.~7, pp.~4860--4880, Jul.~2026.
 


\bibitem{PR2024}
R.~Premlal and B.~S.~Rajan,
``On function-correcting codes,''
in \textit{Proc. 2024 IEEE Information Theory Workshop (ITW)}, Shenzhen, China, 2024, pp.~603--608.
 
\bibitem{GXZZ2025}
G.~Ge, Z.~Xu, X.~Zhang, and Y.~Zhang,
``Optimal redundancy of function-correcting codes,''
\textit{arXiv preprint arXiv:2502.16983}, Feb.~2025.

\bibitem{XLC2024}
Q.~Xia, H.~Liu, and B.~Chen,
``Function-correcting codes for symbol-pair read channels,''
\textit{IEEE Trans. Inf. Theory}, vol.~70, no.~11, pp.~7807--7819, Nov.~2024.
 
\bibitem{SSY2025}
A.~Singh, A.~K.~Singh, and E.~Yaakobi,
``Function-correcting codes for $b$-symbol read channels,''
\textit{arXiv preprint arXiv:2503.12894}, Mar.~2025.

 \bibitem{RRHH20252}
C.~Rajput, B.~S.~Rajan, R.~Freij-Hollanti, and C.~Hollanti,
``Function-correcting codes for locally bounded functions,''
in \textit{Proc. 2025 IEEE Information Theory Workshop (ITW)}, Sydney, Australia, 2025, pp.~851--856.


\bibitem{VSA2025}
G.~K.~Verma, A.~Singh, and A.~K.~Singh,
``Function-correcting $b$-symbol codes for locally $(\lambda, \rho, b)$-functions,''
\textit{IEEE Trans. Inf. Theory}, vol.~72, no.~1, pp.~331--341, Jan. 2026.

\bibitem{SR2025}
S.~Sampath and B.~S.~Rajan,
``On Plotkin bound for function-correcting codes for $b$-symbol read channels,''
in \textit{Proc. 2025 IEEE Information Theory Workshop (ITW)}, Sydney, Australia, 2025, pp.~698--703.

\bibitem{LS2025}
H.~Ly and E.~Soljanin,
``On the redundancy of function-correcting codes over finite fields,''
in \textit{Proc. 2025 13th International Symposium on Topics in Coding (ISTC)}, Los Angeles, CA, USA, 2025, pp.~1--5.

\bibitem{RRHH20262}
C.~Rajput, B.~S.~Rajan, R.~Freij-Hollanti, and C.~Hollanti,
``Non-existence of some function-correcting codes with data protection,''
accepted for presentation in \emph{Proc. IEEE Int. Symp. Inf. Theory (ISIT)}, 2026.
Available: arXiv:2603.01049.

\bibitem{HH2026}
H.~Liu and H.~Liu,
``Function-correcting codes with homogeneous distance,''
\textit{Finite Fields Their Appl.}, vol.~112, p.~102791, 2026.

\bibitem{VS2025}
G.~K.~Verma and A.~K.~Singh,
``On function-correcting codes in the Lee metric,''
\textit{arXiv preprint arXiv:2507.17654}, Jul.~2025.

\bibitem{VS2026}
G.~K.~Verma and A.~K.~Singh,
``Function-correcting codes for linear and locally bounded functions over a finite chain ring,''
\textit{arXiv preprint arXiv:2603.14471}, Mar.~2026.

\bibitem{HUR2025}
K.~Hareesh, N.~T.~Rashid~Ummer, and B.~S.~Rajan,
``Plotkin-like bound and explicit function-correcting code constructions for Lee metric channels,''
\textit{arXiv preprint arXiv:2508.01702}, Aug.~2025.

\bibitem{SS2025}
A.~Singh and A.~K.~Singh,
``Function-correcting codes for insertion-deletion channel,''
\textit{arXiv preprint arXiv:2512.07243}, Dec.~2025.

\bibitem{PBMR2026}
R.~Pandey, S.~Bajpai, A.~A.~Mahesh, and B.~S.~Rajan,
``Function correcting codes for maximally-unbalanced Boolean functions,''
\textit{arXiv preprint arXiv:2601.10135}, Jan.~2026.

\bibitem{DMKPR2026}
S.~S.~Durgi, A.~A.~Mahesh, A.~Kumari, R.~Pandey, and B.~S.~Rajan,
``Function-correcting codes with optimal data protection for Hamming code membership,''
\textit{arXiv preprint arXiv:2602.21932}, Feb.~2026.


\bibitem{D1973} P. Delsarte, ``An algebraic approach to the association schemes of coding theory,'' \emph{Philips Res. Rep. Suppl.}, no. 10, 1973.


\bibitem{MS1977} F. J. MacWilliams and N. J. A. Sloane, \emph{The Theory of Error-Correcting Codes}. Amsterdam, The Netherlands: North-Holland, 1977.

\bibitem{Aumann1976}
R.~J.~Aumann,
``Agreeing to disagree,''
\textit{Ann. Stat.}, vol.~4, no.~6, pp.~1236--1239, Nov.~1976.

\bibitem{Quint2014}
D.~Quint,
``Lecture 3: Common knowledge and agreeing to disagree,''
lecture notes, Univ. Wisconsin--Madison, Madison, WI, USA, 2014.
[Online]. Available: \url{https://users.ssc.wisc.edu/~dquint/econ698/lecture%203.pdf}






\bibitem{vL1999} J. H. van Lint, \emph{Introduction to Coding Theory}, 3rd ed. Berlin, Germany: Springer, 1999.

\end{thebibliography}
\end{document}